\documentclass[twocolumn]{bmcart}

\usepackage[utf8]{inputenc} 

\usepackage{xspace}
\usepackage{amsmath,amssymb,amsthm,latexsym,graphicx}
\usepackage[normalem]{ulem}
\usepackage{listings}

\startlocaldefs

\newtheorem{theorem}{Theorem}
\newtheorem{coro}[theorem]{Corollary}

\newtheorem{lemma}[theorem]{Lemma}
\newtheorem{prop}[theorem]{Proposition}
\theoremstyle{definition}
\newtheorem{defi}[theorem]{Definition}

\newtheorem*{alg}{Algorithm}
\newtheorem{example}[theorem]{Example}

\newcommand{\RR}{{\mathbb R}}
\newcommand{\PP}{{\mathbb P}}

\newcommand{\s}{{\mathcal S}}

\newcommand{\bc}{\begin{center}}
\newcommand{\ec}{\end{center}}
\newcommand{\be}{\begin{enumerate}}
\newcommand{\ee}{\end{enumerate}}
\newcommand{\bi}{\begin{itemize}}
\newcommand{\ei}{\end{itemize}}
\usepackage{amsmath}

\usepackage{listings}

\definecolor{OliveGreen}{rgb}{.3,.5,.2}
\definecolor{MidnightBlue}{rgb}{.3,.4,.6}
\usepackage[utf8]{inputenc}

\usepackage{pdfsync}

\definecolor{codegreen}{rgb}{0,0.6,0}
\definecolor{codegray}{rgb}{0.5,0.5,0.5}
\definecolor{codepurple}{rgb}{0.58,0,0.82}
\definecolor{backcolour}{rgb}{0.95,0.95,0.92}

\newcommand{\grove}{\mathcal G(N^-)} 
\newcommand{\supp}{\operatorname{supp}}
\newcommand{\LSA}{\operatorname{LSA}}
\newcommand{\CNA}{Circular Network Algorithm}
\newcommand{\tc}{\text{:}}
\newcommand{\qc}{qcCF\xspace}
\newcommand{\qcs}{qcCFs\xspace}

\endlocaldefs

\begin{document}

\begin{frontmatter}

\begin{fmbox}
\dochead{Research}

\title{NANUQ: A method for inferring species networks from gene trees under the coalescent model}

\author[
addressref={aff1},                   
]{\inits{ESA}\fnm{Elizabeth S} \snm{Allman}} 
\author[
addressref={aff1},
]{\inits{HB}\fnm{Hector} \snm{Ba\~nos}}
\author[
addressref={aff1},                   
corref={aff1},                       
email={j.rhodes@alaska.edu}   
]{\inits{JAR}\fnm{John A} \snm{Rhodes}}

\address[id=aff1]{
  \orgname{Department of of Mathematics and Statistics, University of Alaska Fairbanks}, 
  \city{Fairbanks, AK},                              
  \cny{USA}                                    
}

\begin{artnotes}
\end{artnotes}

\end{fmbox}


\begin{abstractbox}

\begin{abstract} 
  Species networks generalize the notion of species trees to allow for
  hybridization or other lateral gene transfer.  Under the Network
  Multispecies Coalescent Model, individual gene trees arising from a
  network can have any topology, but arise with frequencies dependent
  on the network structure and numerical parameters.  We propose a new
  algorithm for statistical inference of a level-1 species network
  under this model, from data consisting of gene tree topologies, and
  provide the theoretical justification for it. The algorithm is based
  in an analysis of quartets displayed on gene trees, combining
  several statistical hypothesis tests with combinatorial ideas such
  as a quartet-based intertaxon distance appropriate to networks, the
  NeighborNet algorithm for circular split systems, and the Circular
  Network algorithm for constructing a splits graph.
\end{abstract}



\end{abstractbox}
%

\end{frontmatter}

 
\section*{Introduction}

In this paper we provide the theory supporting a new, statistically
consistent method of inferring most topological features of a level-1
hybridization network under the network multispecies coalescent (NMSC)
model. The method uses as data a collection of unrooted topological
gene trees, which may themselves have been inferred from sequences.

Unlike pseudo-likelihood methods \cite{Solis-Lemus2016,Nakhleh2015},
our method does not require an assumed limit on the number of
hybridization events in the network, nor does it involve a
time-intensive search over the space of possible networks. Instead, it
computes a certain distance between taxa which, under ideal
circumstances, corresponds to a circular split system. When the
expected distance is processed through particular algorithms to
produce a splits graph, interpretation rules allow one to read off
network information. The total theoretical running time of the
algorithm is $\mathcal O(n^4m)$ for an input of $m$ binary gene trees
on $n$ taxa, making it computationally feasible when $n$ has moderate
size.

While we illustrate the method's utility through several examples with
simulated and empirical data, our focus in this work is on providing
its theoretical basis. This draws on a number of independent research
works, but also requires new results on the nature of the splits
graphs that are produced under ideal circumstances.

We call this new method the \underline{N}etwork inference
\underline{A}lgorithm via \underline{N}eighbourNet \underline{U}sing
\underline{Q}uartet distance, or by the acronym NANUQ\footnote{The
  word for ``polar bear'' in I\~nupiaq and other Inuit languages,
  pronounced and sometimes written as `Nanook'.}.  It involves the
following steps, applied to a collection of unrooted gene tree
topologies assumed to have arisen under the NMSC on an unknown binary
level-1 network: \be
\item[(a)] \label{step:CFs} For each subset of 4 taxa, determine the
  empirical quartet counts from the gene trees, which will reflect
  possible cycles on the network, as shown in
  \cite{Solis-Lemus2016,Banos2018}.
\item[(b)] \label{step:stat} Apply a statistical hypothesis test to
  these counts, as in \cite{Mitchell2018}, to judge evidence as to
  whether the quartet species network displays a 4-cycle.
\item[(c)] \label{step:qdist} Use the test results on quartets to
  construct a network quartet distance between taxa, extending the
  ideas of \cite{Rhodes2019}.
\item[(d)] \label{step:NN} Apply the NeighborNet \cite{Bryant2004} and
  Circular Network algorithms \cite{Dress2004} to construct a splits
  graph from the quartet distance.
\item[(e)] \label{step:interp} Interpret the abstract network produced
  in the previous step by certain rules developed in this paper to
  infer most topological features of the unknown network.  \ee All
  steps but the last have been fully automated; in R for the steps
  (a--c), and SplitsTree4 \cite{Huson2006} for step (d).  While it is
  conceivable the last step could be as well, there are advantages to
  not doing so until more experience with the method has accumulated.
  For instance, some data sets may not support a hypothesis of
  evolution on a level-1 hybridization network, and a human
  interpretation of both the hypothesis test results of step (b) and
  the SplitsTree4 output of step (e) may suggest this.  Simply
  returning a hybridization network most in accord with the output
  might be misleading if poor model fit is ignored.

\medskip

NANUQ offers several
important advantages over other network inference methods we know of.
In particular, it can indicate poor model fit to the level-1 NMSC and, in the
case of reasonable fit, indicate the number of hybridization events
without conducting a time-consuming search. In contrast,
pseudo-likelihood methods, which can be used for network inference
\cite{Solis-Lemus2016,Nakhleh2015}, are known broadly to be poor for
judging model fit, though often perform well for inference.  However,
NANUQ only gives information on network topology, whereas
pseudo-likelihood can be used to obtain metric information as well.
We thus view NANUQ as complementary to existing approaches.

Several recent works \cite{Wen2018,Zhang2018} have taken a Bayesian
approach to inference of species networks from genetic sequence data,
to obtain a joint posterior on both species networks and gene
trees. As attractive as one might find this as a conceptual approach,
it produces a formidable computational challenge for data sets with
many taxa or gene trees. Indeed, the largest analyses in these works
are quite small, involving only 7 taxa and 106 gene trees from a yeast
data set which we also analyze. The alternative approaches offered by
NANUQ and the pseudo-likelihood algorithms easily handle much larger
data sets, with thousands of genes, as have already been assembled by
researchers.

\medskip

We note that NANUQ's use of a splits graph is the first instance, to
our knowledge, of such a graph being given a firm model-based interpretation as
supporting a biological process underlying a data set. Splits graphs
are generally viewed as exploratory devices for judging the extent to
which a data set is ``tree-like,'' and authors often warn against
interpreting them as supporting any particular biological mechanism
\cite{HusonRuppScorn}. We fully agree with this general statement;
only in the framework of our multi-step algorithm do we claim that an
interpretation of support for a hybridization network is justified by
theory. While an earlier step in this direction was taken by
\cite{Huson2005}, that work assumed no coalescent process modeling
incomplete lineage sorting was involved in the formation of gene
trees, and provided a less detailed description of the form of a
splits graph than is given here.

\medskip

The theory we present is based on consideration of the quartets
displayed on a collection of gene trees arising under the NMSC, but it
differs in important ways from the more purely combinatorial work,
such as \cite{Gambette2012}, on undirected networks of level-1 and
higher. First, we crucially focus on unrooted phylogenetic networks in
the sense of \cite{Solis-Lemus2016, Banos2018}, which retain the direction of hybrid
edges from the rooted species network underlying the biological model,
rather than fully undirected networks of \cite{Gambette2012}. This
leads to a different notion of the trees and quartets displayed on a
network, and of the set of splits we associate to a network.  Second,
unlike most purely combinatorial studies, our algorithm takes into
account that due to the coalescent process some gene trees will
display quartets inconsistent with the species network. Nonetheless
NANUQ provides a means of determining, up to statistical inference
error, which quartets are displayed on the network. Third, if these
quartets are known exactly, we are able to recover not only the
undirected version of the network (modulo contraction of 2- and
3-cycles) but also directions of hybrid edges in cycles of size 5 or
larger.

\medskip

This paper proceeds as follows: We first outline and develop theory
behind the NANUQ algorithm in a purely theoretical setting. This
constitutes the majority of the work. We then more carefully outline
the algorithm for data analysis, and conclude with a few examples of
network inference.

In more detail, the theoretical portion of this work first formally
defines the type of phylogenetic networks which underly our model, as
well as unrooted semidirected networks induced from them. While this precise
notion of unrooted network appeared in \cite{Banos2018}, it is not
standard to the literature, yet it is essential to our work. Briefly
recalling the network multispecies coalescent model (NMSC) and the
notion of a quartet concordance factor (CF), we summarize results of
\cite{Solis-Lemus2016} and \cite{Banos2018} indicating how these
concordance factors reflect quartet network topology, and provide a
new analysis indicating the extent to which one can avoid the one
important case of ambiguity in interpreting CFs.  After reviewing
terminology for split systems, we then define a split system
associated to an unrooted semidirected level-1 network.  This is used
to define a new quartet intertaxon distance for a level-1 topological
network, which can be computed from quartet information alone.  We
then investigate the splits graph computed from the quartet distance
of a binary level-1 network.  This requires establishing some new
theoretical results which enable us to directly relate the form of a
level-1 hybridization network to the form of the splits graph found
from its network quartet distance.

Finally, we present our algorithm in full, making use of all the
theory above, as well as hypothesis testing using CFs as developed in
\cite{Mitchell2018}, and the NeighborNet \cite{Bryant2004}
and Circular Network \cite{Dress2004} algorithms as implemented in SplitsTree4
\cite{Huson2006}. We give a running time analysis for NANUQ and
establish its statistical consistency.  As our primary goal in this
paper is to provide the theoretical background to our algorithm, we
conclude with a minimial set of example analyses, using both simulated
and biological data.  A later work, directed at empiricists, will
focus further on NANUQ's performance in data analysis.


\section*{Phylogenetic Networks }\label{sec:networks}
\subsection*{Rooted and unrooted phylogenetic networks}\label{sec:def}
We begin by establishing terminology for phylogenetic
networks. Throughout, $X=\{x_1,x_2,\dots ,x_n\}$ denotes a fixed set
of taxa.

Our focus is on an explicit network \cite{HusonRuppScorn}, that can be
interpreted as providing an evolutionary history of species
relationships, including hybridization or other forms of lateral gene
transfer that occur at discrete moments in time.
\begin{defi}[\cite{Banos2018,Steel2016}]\label{def:network} 
  A \emph{topological binary rooted phylogenetic network} $N^+$ on
  taxon set $X$ is a connected directed acyclic graph with vertices
  $V$ and edges $E$, where $V$ is the disjoint union $V = \{r\} \sqcup
  V_L \sqcup V_H \sqcup V_T$ and $E$ is the disjoint union $E = E_H
  \sqcup E_T$, together with a bijective leaf-labeling function $f :
  V_L \to X$ with the following characteristics:
	\begin{itemize}
		\item[1.]The  \emph{root} $r$ has indegree 0 and outdegree 2.
		\item[2.] A  \emph{leaf} $v \in V_L$ has indegree 1 and outdegree 0.
		\item[3.] A   \emph{tree node} $v\in  V_T$ has indegree 1 and outdegree 2.
		\item[4.] A  \emph{hybrid node} $v\in  V_H$ has indegree 2 and outdegree 1.
		\item[5.] A  \emph{hybrid edge} $e \in E_H$ is an edge whose child is a hybrid node.
		\item[6.] A  \emph{tree edge} $e \in E_T$ is an edge whose child is a tree node or a leaf.	 
	\end{itemize}
\end{defi}

\begin{defi}
  Let $N^+$ be a topological binary rooted phylogenetic network. A
  \emph{metric for $N^+$} is a pair $(\lambda, \gamma)$, where
  $\lambda:E\to \mathbb{R}^{\ge 0}$ assigns \emph{edge lengths} and
  $\gamma:E_H\to (0,1)$ assigns \emph{hybridization parameters}
  satisfying
	\begin{itemize}
	\item[1.] $\lambda(e) >0$ for $e\in E_T$,
	\item[2.] $\gamma(e_1)+\gamma(e_2)=1$ whenever $e_1,e_2\in E_H$ have the same hybrid-node child. 
	\end{itemize}
	If $(\lambda, \gamma)$ is a metric for $N^+$, then we refer to
        $(N^+,(\lambda, \gamma))$ as a \emph{metric binary rooted
          phylogenetic network}.
\end{defi}
While the idea of unrooting a tree is simple, unrooting a network is
more subtle. For example, it may not be clear how to proceed when the
two edges incident to the root have the same child. We follow
\cite{Banos2018} in elucidating this concept.

In a directed network, we say that a node $v$ is \emph{above} a node
$u$, and $u$ is \emph{below} $v$, if there exists a non-empty directed
path in $N^+$ from $v$ to $u$. We also say that an edge with parent
node $x$ and child $y$ is \emph{above} (\emph{below}) a node $v$ if
$y$ is above or equal to $v$ ($x$ is below or equal to $v$).

\begin{defi}[\cite{Steel2016}]\label{def:LSA}
  Let $N^+$ be a (metric or topological) binary rooted phylogenetic
  network on $X$ and $Z\subseteq X$. Let $D$ be the set of nodes which
  lie on every directed path from the root $r$ of $N^+$ to any $z\in
  Z$. Then the \emph{lowest stable ancestor of $Z$ on $N^+$}, denoted
  $\LSA(Z)$, is the unique node $v\in D$ such that $v$ is below all
  $u\in D$, $u\neq v$.
\end{defi}
The lowest stable ancestor is a generalization (though not the only
one) on a network of the concept of most recent common ancestor on a
tree.

If $z$ is a degree two node on a semidirected graph, with nodes $x$
and $y$ adjacent to $z$, then by \emph{suppressing} $z$ we mean
deleting $z$ and its incident edges, and introducing a new edge from
$x$ to $y$. If the deleted edges formed a semidirected path, we direct
this new edge consistently with that path; otherwise the new edge is
undirected.

\begin{defi}\label{def:unrooted}	
  Let $N^+$ be a binary topological rooted phylogenetic network on a
  set of taxa $X$. Then $N^-$, the \emph{topological unrooted
    phylogenetic network induced from $N^+$}, is the semidirected
  network obtained by
	\begin{enumerate}
	\item[1.] deleting all edges and nodes above  $\LSA(X)$, 
	\item[2.]  undirecting  all tree edges, and
	\item[3.]  suppressing  $\LSA(X)$.
\end{enumerate}
\end{defi}
If $N^+$ has a metric structure, then $N^-$ inherits one in an obvious
way. Edge lengths on $N^-$ are the sum of conjoined edge lengths in
$N^+$, and hybridization parameters are the same as those on $N^+$.

Note that in some other phylogenetic works the term ``unrooted
network" is used for a fully undirected network. An unrooted network
in our sense retains directions on hybrid edges, and thus encodes some
information about possible root locations on $N^+$.  Figure
\ref{fig:unroot} depicts a topological binary rooted phylogenetic
network on the left and its induced topological unrooted network on
the right.

\medskip

For simplicity, when we refer to an \emph{unrooted network} $N^-$ in
this paper, either metric or topological, we mean a semidirected
network induced from a rooted binary phylogenetic network $N^+$ as in
Definition \ref{def:unrooted}.  That is, we implicitly assume the
existence of $N^+$.  This is an important convention to keep in mind,
since under the standard graph theoretical definition there are
unrooted networks which are not so induced.
	
 \begin{figure*}
	\includegraphics[width=.6\textwidth]{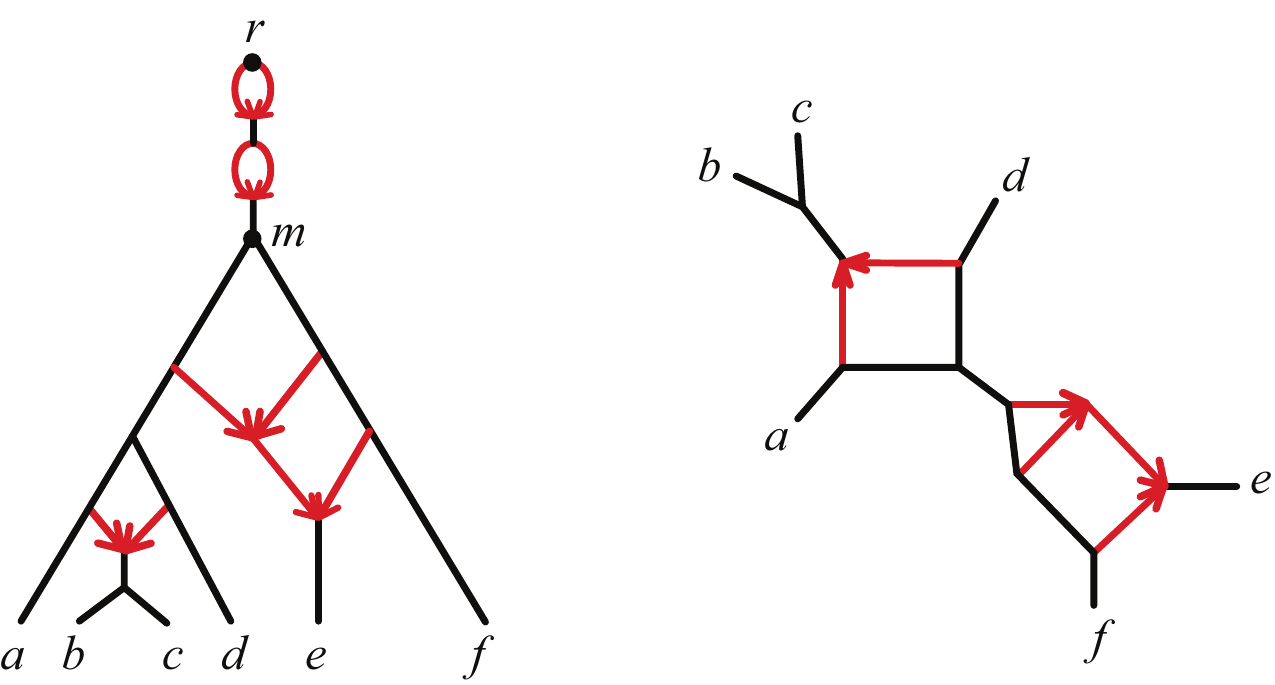}
	\caption{(L) A rooted phylogenetic network $N^+$ with root $r$
          and lowest stable ancestor $m$, and (R) the unrooted network
          $N^-$ induced from $N^+$. }\label{fig:unroot}
\end{figure*}

Since an unrooted network retains some directed edges, a useful
definition of an induced quartet network is more elaborate than the
analog for a tree.  Recall that a \emph{trek} between vertices $x,y$
on a network is the union of semidirected paths from some vertex $v$
to $x$ and from $v$ to $y$. A trek is \emph{simple} if the two paths
intersect only at $v$.
 
 \begin{defi}	
   Let $N^-$ be a unrooted network on $X$, and let $a,b,c,d\in X$. The
   \emph{induced quartet network} $Q_{abcd}$ is the unrooted network
   obtained by
	\begin{enumerate}
	\item[1.] keeping only the edges in simple treks between pairs of elements of $\{a,b,c,d\}$, and  then 
	\item[2.] suppressing all degree two nodes. 
	\end{enumerate}
	 \end{defi}

         In the case that $N^-$ is a metric network, the quartet
         network $Q_{abcd}$ inherits a metric structure in a natural
         way: Noting that any hybrid edge $e$ in $Q_{abcd}$ arises
         from a single hybrid edge $\tilde e$ of $N^-$ possibly
         conjoined with several tree edges, we set the hybridization
         parameter for $e$ equal to that for $\tilde e$.  Edge lengths
         in $Q_{abcd}$ are simply sums of lengths of conjoined edges
         from $N^-$.
 
         Figure \ref{fig:quartet} shows several quartet networks
         induced from the unrooted network in Figure \ref{fig:unroot}.
 
 \begin{figure*}
	 \includegraphics[width=.7\textwidth]{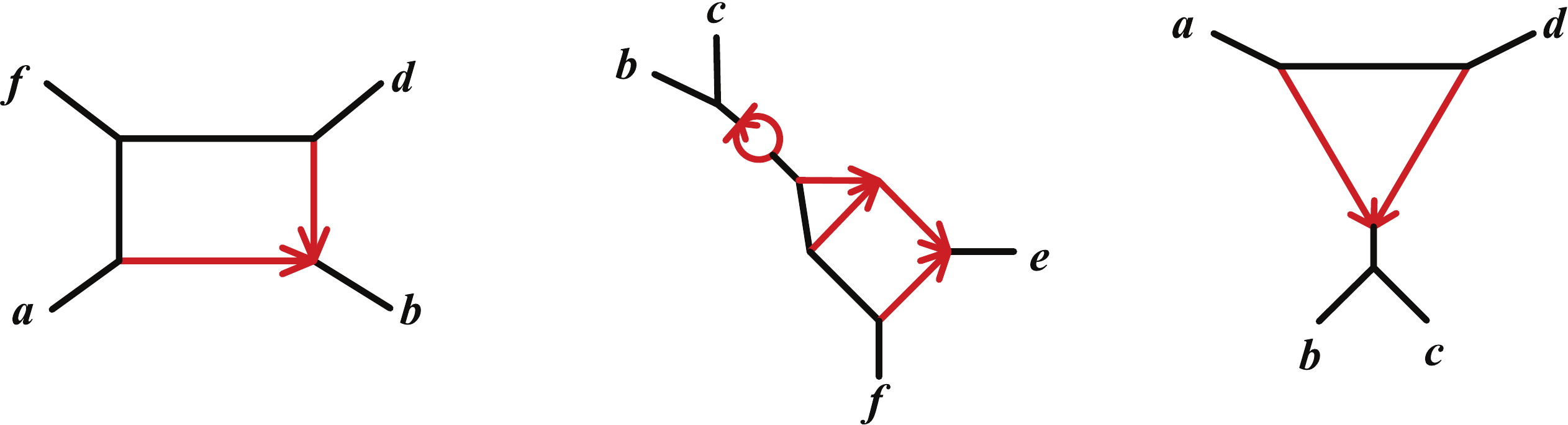}
 	\caption{Three quartet networks, $Q_{abdf}$, $Q_{bcef}$, and $Q_{abcd}$,
	 induced from the unrooted network $N^-$ of Figure \ref{fig:unroot} (R).
	}\label{fig:quartet}
 \end{figure*}
 
 \medskip
 
Finally, most of our results are established only for a subclass of phylogenetic networks 
exhibiting a \emph{level-1} structure.  The definition we give is not the standard one for 
level-1 (e.g., \cite{Steel2016}), but it is equivalent for binary directed networks \cite{Rossello2009}. 
We also use our notion of level-1 for the unrooted networks in this paper,
where the directions of hybrid edges are preserved. 
\begin{defi}\label{level1}Let $N$ be a (rooted or unrooted) binary topological  network. If no two 
cycles in the undirected graph of $N$ share a vertex, then $N$ is \emph{level-1}. \
\end{defi}


\section*{The Network Multispecies Coalescent Model and Quartet Concordance Factors}\label{sec:modelCF}

The \textit{multi-species coalescent model} (MSC)
\cite{Pamilo1988,MSCpaper2009} is the standard probabilistic model of
incomplete lineage sorting, by which gene trees, showing direct
ancestral relationships, form within species trees composed of
multi-individual populations.  It traces, backwards in time, the
lineages of a finite set of individual copies of a gene, sampled from
different extant species, as they \emph{coalesce} at common ancestral
individuals.

The \textit{network multi-species coalescent  model} (NMSC)  \cite{Meng2009,Nakhleh2012,Nakhleh2016} is a 
generalization of the MSC, which allows a finite number of  hybridization events, or other discrete horizontal 
gene transfer events, between populations.  Its parameters are captured by a metric, rooted phylogenetic network, 
assumed to be binary, as defined above. Branch lengths are given in coalescent units, so that
the rate of coalescence between two lineages is 1.
At a hybrid node in the network, a gene lineage may pass into either of two ancestral populations, with probabilities 
given by the hybridization parameters $\gamma, 1-\gamma$ for that node. 
This differs from other generalizations of the MSC, such as 
those built on a structured coalescent, where genes may switch populations continuously  over an interval in time.

\subsection*{Quartet concordance factors} The NSMC model is often used to obtain the probability (or density) of 
observing a specific gene tree (metric or topological, rooted or unrooted) in a species network.  The NANUQ
algorithm focuses on summaries of gene trees; that is, that a species network produces various gene tree quartets 
(unrooted topological gene trees on 4-taxa) in parameter-dependent frequencies under the NMSC. 
The study of these probabilities, and their use for network inference, was pioneered in \cite{Solis-Lemus2016}, 
with further work in \cite{Banos2018}.
A key concept is that of a quartet concordance factor, whose definition we recall.

A binary unrooted topological tree on four taxa $a,b,c,d$ is called a \emph{quartet}, denoted as $ab|cd$ if deletion 
of its internal edge gives a connected component $\{a, b\}$.  When $n \ge 4$, 
an $n$-taxon tree \emph{displays} a quartet $ab|cd$ if the induced unrooted tree on the four taxa is $ab|cd$.

\begin{defi}
  Let $N^+$ be a metric rooted network on a taxon set $X$, and
  $A,B,C,D$ genes sampled from individuals in species $a,b,c,d\in X$
  respectively.  Given a gene quartet $AB|CD$, the \emph{concordance
    factor} $CF_{AB|CD} = CF_{AB|CD} (N^+)$ is the probability under
  the NMSC on $N^+$ that a gene tree displays the quartet $AB|CD$. The
  \emph{concordance factor} $CF_{abcd} = CF_{abcd}(N^+)$ is the
  ordered triple
	$$CF_{abcd}=(CF_{AB|CD},CF_{AC|BD},CF_{AD|BC})$$
of concordance factors of each quartet on the taxa $a,b,c,d$.
\end{defi}

When there is no ambiguity, such as when we have a fixed rooted metric
network $N^+$ in mind, we denote the concordance factor simply by
$CF_{abcd}$.  Similarly, when $a,b,c,d$ are clear from context (e.g.,
if $N^+$ has only four taxa), we write $CF$ for $CF_{abcd}$.  Also,
while the language of `concordance factor' is sometimes used for both
theoretical values and empirical estimates, in this work we use this
term exclusively for the expected values, being careful to
refer to `estimators of CFs,' or `empirical CFs,' when these are
computed from data.

\medskip

As established  in \cite{Solis-Lemus2016, Banos2018},  the concordance factors for a level-1 network $N^+$ 
depend only on the unrooted network $N^-$, and, more precisely, $CF_{abcd}$ depends only on the quartet network 
$Q_{abcd}$ induced from $N^-$.  Significantly, these concordance factors carry information about what 4-taxon
substructures might be on that network.   For instance, if four taxa $a$, $b$, $c$, $d$ are related by the tree
$ab|cd$ on $N^-$, then under the NMSC the concordance factors satisfy
$CF_{AB|CD} > CF_{AC|BD} = CF_{AD|BC}$.   To explain what information $CF_{abcd}$ contains about
cycle structure on $N^+$, we quickly review some terminology and results from these works.  

By an $m_k$-cycle in a level-1 network we mean an $m$-cycle with
exactly $k$ taxa descended from its unique hybrid node.  In a level-1
quartet network, there are exactly 6 types of cycles that may appear:
$2_1$-, $2_2$-, \mbox{$2_3$-}, $3_1$-, $3_2$-, and $4_1$-cycles which
are depicted in Figure \ref{fig:cycles}. When considering level-1
quartet networks, there are restrictions on the number and types of
cycles that may occur simultaneously.  For example, $Q_{abcd}$ might
have a $4_1$-cycle or a $3_2$-cycle, but not both.

We next classify concordance factors $CF_{abcd}$ depending on the magnitude of its entries.

 \begin{figure*}
 \begin{center}
 \includegraphics[scale=.21 ]{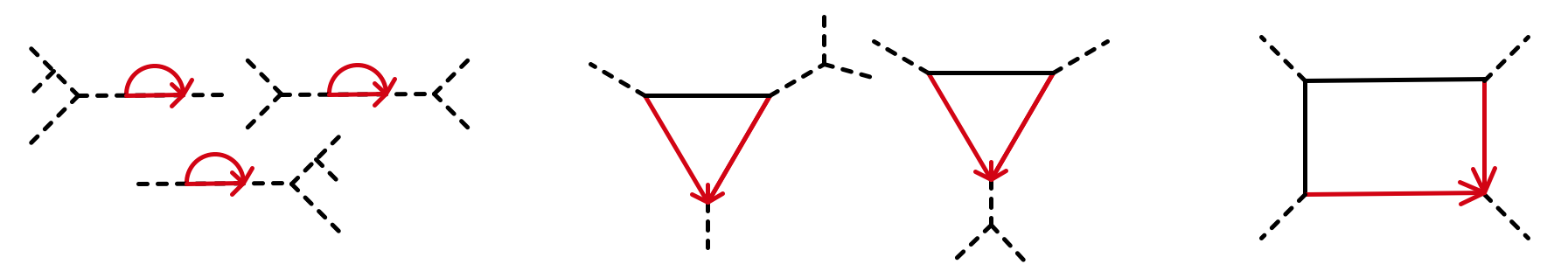}
 \caption{Cycles in a level-1 quartet network are classified as type $m_k$ if they have $m$ edges and 
 $k$ descendants of the hybrid node. The only cycles possible in a level-1 quartet network are of (L) type  
 $2_1$, $2_2$, and $2_3$; (C) type $3_1$ and $3_2$; and (R) type $4_1$. 
 The dashed lines represent subgraphs that may contain other $m_k$ cycles for $m=2,3$.
 }\label{fig:cycles}
 \end{center}
 \end{figure*}
 
\begin{defi} If the two smallest entries of the concordance factor $CF = CF_{abcd}$ are equal, then $CF$ is
said to be \emph{tree-like}.
If a tree-like CF has a unique largest entry, without loss of generality $CF_{AB|CD}$, then 
$CF$ \emph{supports} the quartet $ab|cd$. 
If $CF = (1/3,\, 1/3, \,1/3)$, then it supports all three quartets.
\end{defi}

This terminology is motivated by the fact that if a concordance factor
CF arises from the NMSC on a species \emph{tree}, then CF is
tree-like, and its largest entry indicates the quartet species tree
topology \cite{Allman2011}.  However, as was first shown in
\cite{Solis-Lemus2016}, certain types of non-tree networks also
produce tree-like CFs under the NMSC.

Viewing CF as a point in the probability simplex $\Delta_2=\{(x_1,x_2,x_3))\mid x_i\ge 0,\sum x_i=1\}$, 
as in Figure \ref{fig:simplex} (L), the tree-like CFs form 
3 line segments radiating from the central point $(1/3,1/3,1/3)$ to the vertices. 
With the ordering $$CF_{abcd}=(CF_{AB|CD},CF_{AC|BD},CF_{AD|BC}),$$
the diagonal segment leading to $(1,0,0)$ comprises those CFs supporting ${ab|cd}$, 
the segment leading to $(0,1,0)$ comprises those supporting $ac|bd$, and the vertical segment
leading to $(0,0,1)$ comprises those supporting $ad|bc$.

The next proposition summarizes several results from \cite{Banos2018}.  By 
 \emph{contraction} of a cycle, we mean the removal of its edges followed by 
 the identification of all vertices in it.

\begin{prop}\label{prop:CF}   Let $N^+$ be a level-1 binary quartet network and $N^-_c$ 
the network obtained from $N^-$ by contracting all $2$- and 3-cycles and then suppressing degree 2 nodes.

\begin{enumerate}
	
\item[1.] If $N^-$ has no cycle of type $4_1$ or $3_2$, then its concordance factor CF is tree-like, 
and supports the quartet  $N^-_c$. 
That is, if $N^-_c= ab|cd$, then
$$CF_{AB|CD}>CF_{AC|BD}=CF_{AD|BC}.$$

\item[2.] If $N^-$ has a $3_2$-cycle, then its concordance factor CF may or may not be tree-like.
In particular, CF is on the extended line segment in $\Delta_2$ containing the tree-like 
concordance factors that support the quartet $N^-_c$. 
Specifically, if $N^-_c= ab|cd$, then
$$CF_{AB|CD}\ge 1/6,  \text{  and  }CF_{AC|BD}=CF_{AD|BC},$$
and any such tree-like CF supports $ab|cd$.

\item[3.] If $N^-$ has a $4_1$-cycle, then its concordance factor CF is not tree-like, 
and if $N^-_c$ displays a 4-cycle joining taxa in circular order $a,b,c,d$,
then
$$CF_{AB|CD}>CF_{AC|BD}  \text{  and  }CF_{AD|BC}>CF_{AC|BD}.$$
\end{enumerate}
\end{prop}

\begin{figure*}
\begin{center}
\includegraphics[width=.85\textwidth]{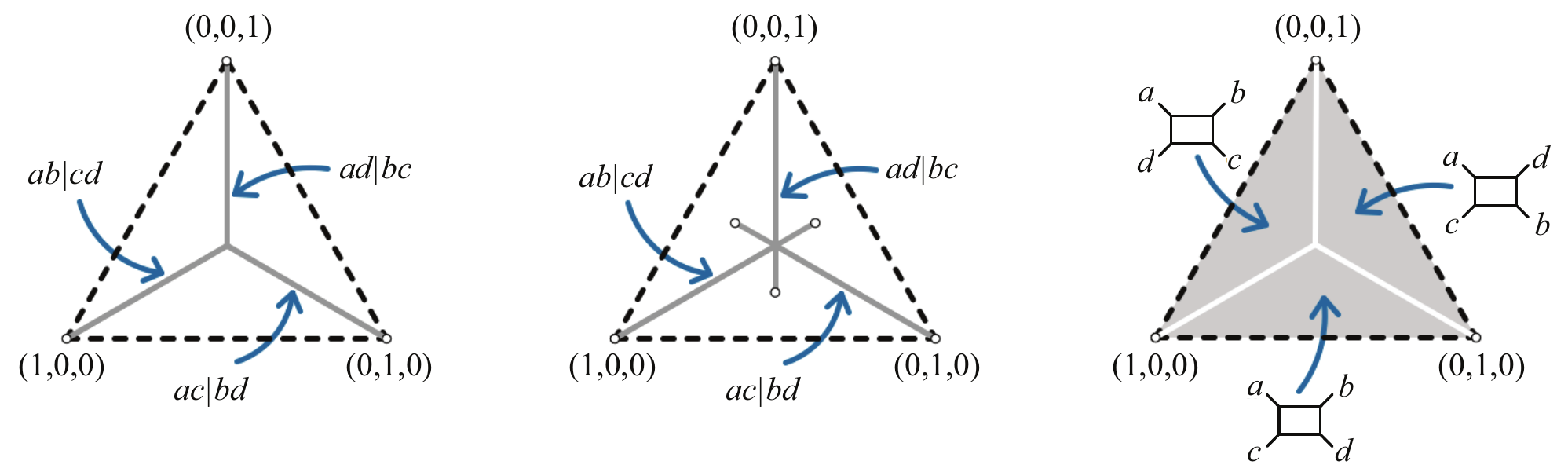} 
\caption{Planar projections of the simplex $\Delta_2$ showing types 
of concordance factors for networks $N^-_c$ of Proposition \ref{prop:CF}.
(L) Gray line segments represent
tree-like CFs that arise from quartet networks with no $3_2$-cycle and with no $4$-cycle. 
(C) Gray line segments represent  CFs that arise from quartet networks with a $3_2$-cycle. 
(R) Gray shaded areas represent CFs that arise from quartet networks containing a $4$-cycle.  
In all three figures, the topology of $N^-_c$ is marked for the 
appropriate line segments or regions of CFs. 
}\label{fig:simplex}
\end{center}\end{figure*}

In Figure \ref{fig:simplex}, we make concrete the proposition's results.  The CFs for binary quartet
networks partition the simplex: $\Delta_2 = \{\text{tree-like CFs}\} \sqcup  \{4_1\text{-cycle CFs}\}$,
with the collection of CFs for $3_2$-cycles meeting both subsets non-trivially.  Notably,
if a quartet network $N^-$ has no $3_2$-cycle, then CFs suffice to determine if $N^-_c$ is a tree or a 4-cycle.
This idea underlies our algorithm, as well as the network identifiability results from
 \cite{Banos2018}.

Indeed, we see from the partition that (in the absence of $3_2$-cycles) the presence
of 2-cycles and $3_1$-cycles has no impact on whether a quartet tree or 4-cycle network is supported.
This observation leads to the non-identifiability of such cycles on a network by the proof method utilized in
\cite{Banos2018}, and prevents NANUQ from detecting them too.   However, since
$2$- and $3_1$-cycles on a large network model `hybridization' between the most closely related populations 
(two that split and then rejoin, or hybridization between two populations which have just split from a common one) 
the inability to infer that such hybridization events occurred by our method may not be too surprising. 
The SNaQ algorithm \cite{Solis-Lemus2016} is likewise unable to detect these, as it too is based
on CFs.

Because concordance factors arising from quartet networks with a
$3_2$-cycle (case 2 of Proposition \ref{prop:CF}) coincide with CFs
for particular parameter choices for $4_1$-cycle networks and
tree-like networks, such CFs must be handled with delicacy.  Clearly,
$3_2$-cycles on quartet networks are not identifiable from CFs, and
therefore will not be reconstructed by the NANUQ algorithm which
focuses only on 4-cycles and tree-like quartet networks.  Because such
$3_2$-cycles will be disregarded, we investigate them more fully next.

A first observation is that for a tree-like CF arising
from a quartet network $N^-$ with a $3_2$-cycle, say with descendants $a$, $b$ of the hybrid node
as in Figure \ref{fig:3c2t},
then $N^-_c$ has topology $ab|cd$.  This is exactly the topology supported by the CF,
when viewed as arising from a particular parameter choice on the $4$-taxon tree
$ab|cd$.  Thus, while determining if the CF arises from a $3_2$-cycle or a tree is not possible, a tree-like
CF always correctly supports the topology of $N_c^-$.

This leaves the question of how `rare' are non-tree-like $3_2$-cycle networks, and what metric
structure on a $3_2$-cycle network might lead to CFs that coincide
with $4_1$-cycle CFs.

\subsection*{$3_2$-cycles}\label{subsec:3_2cycles}

Let $N^-$ be the unrooted quartet network shown in the left of Figure \ref{fig:3c2t}, with branch
length parameters $t_i$ in coalescent units, and hybridization paramter $\gamma$ as shown. 
 With $x_i= e^{-t_i}$ then \cite{Solis-Lemus2016,Banos2018} the quartet concordance factors of $N^-$ are 
 \begin{figure*}
\begin{center}
\includegraphics[width=.8\textwidth]{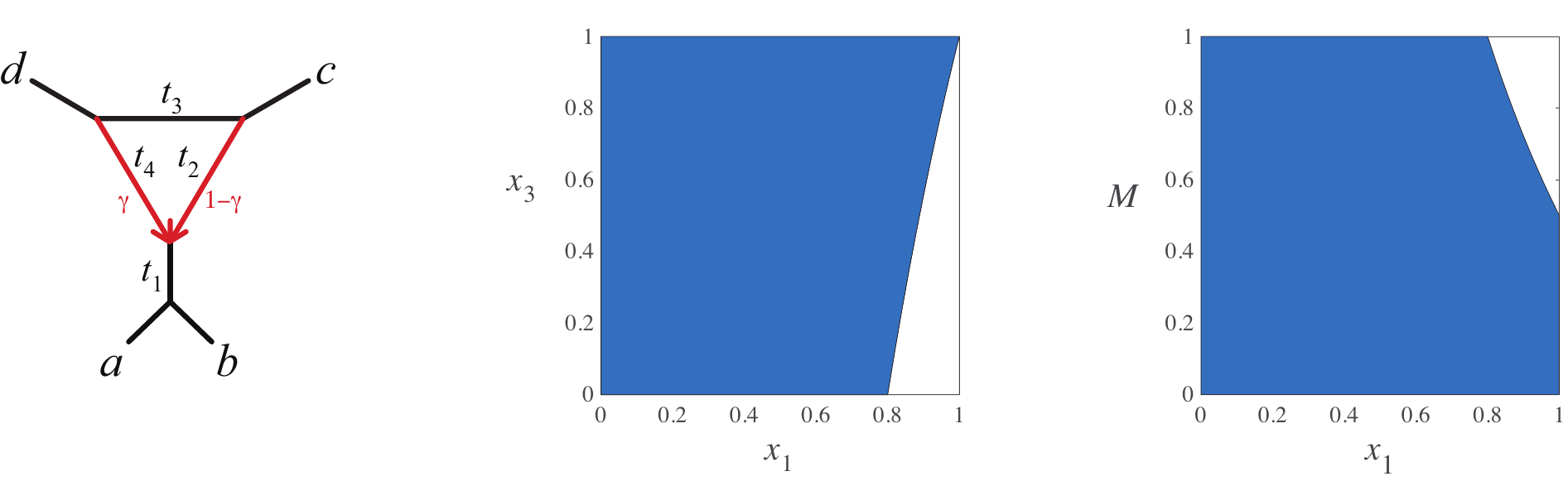}
\end{center}
\caption{(L) NMSC parameters for an induced unrooted quartet $N^-$ with a $3_2$-cycle. 
(C) A region of tree-like parameters $(x_1,x_3)$ on $N^-$ for arbitrary $t_2$, $t_4$, $\gamma$.  (R) A region of tree-like parameters $(x_1, M)$, where $M = \max\{x_2,x_4\}$ for arbitrary $t_3$, $\gamma$. Transformed
parameters are defined by $x_i = e^{-t_i}$.}
\label{fig:3c2t}
\end{figure*}
\begin{align*}
CF_{AB|CD}&=\big(1-\gamma\big)^2\big(1-\frac{2}{3}x_1x_2\big)\\
& \qquad + 2\gamma\big(1-\gamma\big)\big(1-x_1+\frac{1}{3}x_1x_3\big)\notag\\
& \qquad + \gamma^2\big(1-\frac{2}{3}x_1x_4\big),\notag\\
CF_{AC|BD}&=CF_{AD|BC}
=\big(1-\gamma\big)^2\big(\frac{1}{3}x_1x_2\big)\\
& \qquad +\gamma\big(1-\gamma\big)x_1\big(1-\frac{1}{3}x_3\big)\\
& \qquad +\gamma^2\big(\frac{1}{3}x_1x_4\big).\notag 
\end{align*}

We say a choice of parameters $\{t_1,t_2,t_3,t_4,\gamma\}$, or their
transformed versions $x_i$, is \emph{tree-like} if the CF for the
network is tree-like for those parameters.  The set of tree-like
parameters for $N^-$ is a region in the 5-dimensional cube, $0\le
x_1,x_2,x_3,x_4,\gamma \le 1,$ defined by the polynomial inequality
$CF_{AB|CD}\ge CF_{AC|BD}.$

To get a sense of the size of the tree-like region on $N^-$, we
sampled uniformly at random $10^{10}$ points in $[0,1]^5$.  For
untransformed branch length parameters $t_i$, this corresponds to
sampling from an exponential distribution with mean 1.  We computed
that approximately $0.00532$ of the resulting CFs were not tree-like.
In this sense, non-tree-like CFs from $3_2$-cycles are rare.

For additional insight into tree-like parameters on $N^-$, we investigate CFs as functions
of $x_1$ and $x_3$, with $0 < x_2,  x_4, \gamma \le 1$, noting that when $x_2$, $x_4$ achieve their
maximum value of 1, this corresponds to the network with hybrid branch lengths $t_2 = t_4 =  0$.  Concretely,
parameters are tree-like if
\begin{align}\label{eq:CFdiff}
CF&_{AB|CD}- CF_{AC|BD} \\
= & ~1 - x_1 \, \big((1-\gamma)^2 x_2+\gamma (1-\gamma) (3 - x_3) + \gamma^2 x_4\big) \notag\\
\ge &~1 - x_1 \, \big((1-\gamma)^2 +\gamma (1-\gamma) (3 - x_3) + \gamma^2 \big)\notag \\
 = & 1-x_1-\gamma(1-\gamma)x_1(1-x_3)\notag\\
\ge &~1-x_1-\frac14 x_1(1-x_3) \, = \, 1 - \frac{1}{4}x_1(5-x_3) \ge 0.  \notag
\end{align}
Hence parameters are tree-like for any values of $x_2,x_4,\gamma$ when $x_1 \le  4/(5-x_3)$, a region shown in the center of Figure \ref{fig:3c2t}.   
This region has area $4 \ln \frac 54 \approx .89$. More crudely, provided  $x_1 \le 4/5$ 
(that is, $t_1 \ge \log(5/4)\approx 0.2231$ coalescent units), then a tree-like CF results regardless  of all other parameter 
values.  Thus non-tree-like parameters require that $t_1$ be fairly short, causing substantial incomplete lineage sorting. 
For comparison, if the internal branch on a rooted 3-taxon species tree has length
 $t<\log(5/4)$, then fewer than half of the gene trees match the species tree under the MSC.
  
Although this argument assumed the non-existence of  $2_1$-, $2_2$-, and $3_1$-cycles in $N^-$,
a general  level-1 quartet network with a $3_2$-cycle might have cycles of those types.
The result generalizes without difficulty to these more general networks, with $t_1$ the length of the edge 
descended from the $3_2$-hybrid node.
For larger networks, we have the following proposition.

\begin{prop}\label{prop:nobad32}
Suppose $N^+$ is a level-1 network on $n$ taxa and that for each $m_k$-cycle 
with $m \ge 3$ and $k \ge 2$ the branch
descending from the hybrid node has length $t \ge \log(5/4)$.  Then under the NMSC model all 
CFs for induced quartet networks $\mathcal Q$ on $N^-$ are tree-like, except when $\mathcal Q$ has a $4$-cycle.
\end{prop}

Before proving the proposition, note that an $m_k$-cycle in $N^+$ can induce not only a $4_1$-cycle in an induced
quartet network, but also smaller cycles, depending on the particular choice of four taxa.  
For instance, a 4-cycle in the network of Figure \ref{fig:unroot} (L) leads to a $3_2$-cycle
in the induced quartet network on $a, b, c, d$, as shown in Figure \ref{fig:quartet} (R).

\begin{proof}
Choose taxa so that the $3_2$-cycle in $N^-_{abcd}$ and its parameters are named as in Figure \ref{fig:3c2t}.  
Then $t_1 \ge t$ since the edge of length $t_1$ in $N^-_{abcd}$ is made by (possibly) conjoining several edges
in $N^+$, including the one of length $t$.
The argument following equation \eqref{eq:CFdiff} now applies.
\end{proof}

The branch length hypotheses in Proposition \ref{prop:nobad32} are
sufficient, but not necessary, for tree-like CFs in the presence of
$3_2$-cycles.  For instance, if a tree edge $e$ descendant from a
hybrid node in a $3_2$-cycle in $N^+$ is followed by one (or more)
$2$-cycles, then the length requirement on $e$ to produce tree-like
CFs might be shortened.

\smallskip

Focusing again on the quartet network of Figure \ref{fig:3c2t}, we now
investigate transformed branch length parameters $x_2$, $x_4$ on
hybrid edges that lead to tree-like parameter choices.  To this end,
let $M=\max(x_2,x_4) \le 1$.  Then from equation \eqref{eq:CFdiff} for
any $x_3$, $\gamma$, we find
\begin{align*}
CF&_{AB|CD}- CF_{AC|BD}  \\
&\ge 1 - x_1 \, \big((1-\gamma)^2 M+3 \gamma (1-\gamma)  + \gamma^2 M\big)\\
& \ge 1-x_1\left ( \frac{3+2M}4\right ),
\end{align*}
and parameters are tree-like if $M\le \min \left ( \frac 2{x_1}-\frac
  32, 1 \right )$, a region shown in blue in Figure \ref{fig:3c2t}
(R).  Its area is $2 \log(\frac 54 ) + .5 \approx 95\%$ of the $(x_1,
M)$-parameter space shown.  As a special case, if $M\le \frac 12$
(equivalently, $\min\{t_2,t_4\}\geq \log(2)\approx 0.693$ coalescent
units), parameters are tree-like for all choices of $(x_1, x_3,
\gamma)$.

\medskip 

The branch length conditions presented here that rule out
non-tree-like $3_2$-cycles come with a caution, since one might prefer
to avoid \emph{a priori} modeling assumptions on branch lengths.
Nonetheless, our goal has been to suggest that plausible assumptions
can rule out non-tree-like CFs arising from $3_2$-cycles in quartet
networks. Inspection of empirical CFs from a data set may provide
further evidence that no such CFs are involved in a data analaysis
  

\section*{Network split systems and distances}\label{sec:split}

The ability to use quartet CFs to determine whether a quartet network  displays
a 4-cycle can be combined with ideas from \cite{Rhodes2019}
to compute a pairwise distance between taxa on a large $n$-taxon network.   
Indeed, the intertwining of these ideas with that of a weighted circular split system 
is the foundation of the NANUQ algorithm.  
In this section we review the concepts of weighted circular split systems and associated distances,
as needed for our inference method.

\subsection* {Split systems}

We adopt standard terminology concerning splits \cite{Bryant2007}. A \emph{split} $A|B=B|A$ 
of taxa $X$ is a bipartition $X=A\sqcup B$ with $A,B$ non-empty. The subsets $A$, $B$
are called split sets.  The set of all splits of $X$ is 
denoted by $(X)$, and $\mathcal S\subseteq (X)$ is called a \emph{split system} on $X$.

\begin{defi}\label{def:circular split}
A split system $\mathcal S\subseteq (X)$ is 
\emph{circular} if there exists a linear ordering $x_1<...<x_n$ of the elements of $X$ 
such that each split in $S$ has the form $A|B$ with
$$A=\{x_p,x_{p+1},....,x_{q-1},x_q\}$$
for appropriately chosen $1\le p<  q< n$. The ordering of the $x_i$ is
a \emph{circular ordering} for $\mathcal S$.
\end{defi}

A circular ordering for $\s$ is not unique, since it can be modified by
cyclically permuting the $x_i$ (e.g., replaced with $x_2< x_3,<\dots <x_n<x_1$) or by
inversion (replaced with $x_n<x_{n-1}<\dots< x_1$), while remaining a circular ordering for $\s$. 
We treat such variants as the same, without further comment.

\medskip

Given a tree $T$ on $X$, deleting an edge defines a split according to the connected components of the resulting  graph. The set of all such displayed splits is denoted $\mathcal S(T)$, and it is clear from a planar depiction of a tree that $\mathcal S(T)$ is circular.

For a tree, the correspondence between edges and displayed splits allows edge weights to be viewed as split weights, by setting weights of non-displayed splits to 0. This is a special case of a
\emph{weighted split system on $X$}, a map $$\omega:(X)\to \mathbb{R}^{\ge 0}.$$ 

A weighted split system $\omega$ on $X$ induces a distance  function $d_\omega$ on $X$ by
$$ d_\omega(x,y)=\sum_{s\in S_{xy}} \omega(s), $$
where $S_{xy}\subseteq (X)$ is the set of splits separating $x$ and $y$, i.e., splits $A|B$, with $x\in A$ and $y\in B$. Clearly $d_\omega$ is non-negative valued, with $d_\omega(x,x)=0$, $d_\omega(x,y)=d_\omega(y,x)$.

Recall that the support of a weighted split system, 
denoted $\operatorname{supp}(\omega)$, is the set of splits on which $\omega$ is non-zero. 

\begin{defi}\label{def:circular function}
  A weighted split system $\omega$ on $X$ is said to be
  \emph{circular} if $\operatorname{supp}(\omega)$ is circular. A
  distance function $d$ on $X$ is said to be \emph{circular} if
  $d=d_\omega$ for some circular weighted split system $\omega$.
\end{defi}

As pointed out in \cite{Bryant2007}, it follows from  \cite{Dress1992} that a circular distance function $d$ 
uniquely determines the weighted split system $\omega$ such that $d= d_\omega$. 

\subsection* {Splits from unrooted networks}

Our notion of splits associated to a network, and some related
terminology, is not standard, but is essential to this work. In
particular, we focus only on phylogenetic unrooted networks $N^-$ as
in Definition \ref{def:unrooted}, where $N^-$ is induced from a rooted
phylogenetic network and the direction of hybrid edges are retained in
$N^-$.

\begin{defi}\label{def:grove}
  Let $N^-$ be a unrooted network on $X$. An unrooted tree $T$ on $X$
  is \emph {displayed} on $N^-$ if it can be obtained from $N^-$ by
  deleting some edges, including at least one hybrid edge from each
  pair, undirecting remaining hybrid edges, and suppressing degree 2
  nodes. The set of all unrooted topological trees on $X$ displayed on
  $N^-$ is called the \emph{grove} of $N^-$, denoted $\grove$.
\end{defi}

If $N^-$ has an $m$-cycle with $m\ge 4$, then the grove $\grove$ is a
proper subset of the displayed trees on the undirected network $N^u$
underlying $N^-$ as defined in \cite{Gambette2012}.  This is because
$N^u$ is obtained by undirecting the hybrid edges in $N^-$, and there
is additional freedom in the choice of edges to delete in $N^u$ to
obtain its displayed trees: It is not necessary to delete at least one
of the edges from $N^u$ that arose from each pair of hybrid edges in
$N^-$.

\smallskip

If $N^-$ has 2- or  3-cycles, then deleting either hybrid edge in those cycles 
yields trees with the same topology, and hence gives the same elements of $\grove$.  
In contrast, for cycles of size 4 or larger, the trees in $\grove$ vary with
the choice of hybrid edge deleted.  Since we assume that $N^-$ is level-1 
with $k$ cycles of size $\ge 4$, then $|\grove |=2^k$.  

\begin{defi}\label{def:split arrangement}
  For an unrooted network $N^-$, the set of splits $$\mathcal
  S(N^-)=\bigcup_{T\in \grove} \mathcal S(T)$$ is called the
  \emph{(unweighted) split system for $N^-$}. A \emph{weighted split
    system for $N^-$} is any weighted split system with support
  $\mathcal S(N^-)$.
\end{defi}

\medskip

The study of undirected networks in \cite{Gambette2012} provides the
following important theorem, an analog for undirected networks of
Buneman's splits equivalence theorem.

\begin{theorem}[\cite{Gambette2012}]\label{thm:gambette}
Let $S$ be a split system on a set $X$.
Then $S$ is circular if, and only if, there exists an undirected level-1 network $N$ such that $S \subseteq \s(N)$,
the set of all splits of all trees on $X$ displayed on $N$. 
\end{theorem} 

Note that if $N^u$ is the undirected network underlying the unrooted
network $N^-$, then $\s(N^-) \subseteq \s(N^u)$.  As a consequence, we
obtain the following.

\begin{coro}\label{cor:splitcirc} If $N^-$  is  a  level-1 unrooted network, then 
	$\mathcal S(N^-)$  is circular.
\end{coro}


\section*{Quartet Distance for level-1 networks}\label{sec:qdist}

As shown in \cite{Rhodes2019}, a topological tree has a natural metrization tied to the quartets 
displayed on the tree.  Importantly, intertaxon distances from this 
metrization can be computed from the collection of displayed quartets, without having knowledge
of the full tree, giving a means for consistently inferring the tree topology. 
After briefly reviewing these results in the tree setting, we generalize them to the setting of level-1 networks.

 \subsection*{Quartet distance on a tree} 

 For an unrooted binary topological phylogenetic tree $T$ on $X$, any
 internal edge $e$ induces a partition of $X$ into 4 non-empty blocks,
 $X_1$, $X_2$, $X_3$ and $X_4$, where the split associated to $e$ is
 $s_e=X_1\cup X_2|X_3\cup X_4$, and the splits associated to the 4
 adjacent edges have an $X_i$ as one split set. Similarly, a pendant
 edge $e$ to taxon $a$ induces a partition into 3 blocks $X_1$, $X_2$
 and $\{a\}$, where $s_e=\{a\}|X_1\cup X_2$, and the splits associated
 to the 2 edges adjacent to $e$ have an $X_i$ as one split set. The
 quartet weight function $w_T:(X) \to \RR$ is defined as
 	\[ w_T(s)=\begin{cases} 
 |X_1||X_2|+|X_3||X_4| & \text{ if $s=s_e$, $e$ internal}, \\
  |X_1||X_2| & \text{ if $s=s_e$, $e$ pendant}, \\
  0 & \text{ if $s$ is not on $T$.}
 \end{cases}
 \]
 
 This split weight function then induces $d_{w_T}$, the quartet
 distance function on $X$.  This distance is a tree metric, and
 therefore can be used to reconstruct the topological binary $n$-taxon
 tree $T$ by several algorithms.  Significantly, the distance function
 $d_{w_T}$ can computed another way, from the set of quartets
 displayed on $T$, without prior knowledge of the full tree topology.

\begin{theorem} \label{prop:treequart}\cite{Rhodes2019} For any quartet $q$ on taxa in $X$ with 
$|X|=n$, let $\rho_{xy}(q)=1$ if $q=xz|yw$ separates $x,y$, and 0 otherwise. Then for an unrooted 
binary tree $T$ on $X$, and any $x,y\in X$,
\begin{equation}
d_{w_T}(x,y)=2\sum_{\text{$q$ on $T$} }\rho_{xy}(q)+2n-4.
\label{eq:treedis}
\end{equation}
\end{theorem}

\subsection*{Quartet distance on a network}
To generalize Theorem \ref{prop:treequart} to a network,
we begin with a definition.

\begin{defi}\label{def:netsplitwts}
Let $N^-$ be an unrooted network on $X$.  Then  
the \emph{quartet weight function} $\omega_{N^-}$ 
is defined by
$$\omega_{N^-}(s) = \sum_{T\in \grove} w_T(s), $$ 
 where $s \in (X)$ and $w_T(s)$ is the quartet weight function on $T$.
\end{defi}
 
Note that since $\supp(w_T)=\mathcal S(T)$ for each $T$, $\supp (\omega_{N^-})=S(N^-)$.
Thus, by Corollary \ref{cor:splitcirc}, the quartet weight function $\omega_{N^-}$ is a weighted 
circular split system for $N^-$.  Moreover, the induced distance function is easily related to those for the trees
in the grove $\grove$.

\begin{lemma}\label{lem:quartet tress}
 	Let $N^-$ be a level-1 unrooted network  on $X$. Then
 	$$d_{\omega_{N^-}}=\sum_{T\in \grove }d_{w_T}.$$
\end{lemma}
\begin{proof} For $x,y\in X$, let $S_{xy}\subset(X)$ be the set of splits separating $x$ and $y$. Then
\begin{align*}d_{\omega_{N^-}}&(x,y)
=\sum_{s\in S_{xy}}\omega_{N^-}(s)
=\sum_{s\in S_{xy}} \sum_{T\in \grove} w_T(s)\\
&=\sum_{T\in \grove} \sum_{s\in S_{xy}} w_T(s)
=\sum_{T\in \grove }d_{w_T}(x,y).
\end{align*} 
\end{proof}
 
To state a network analog of Theorem \ref{prop:treequart}, we must extend the indicator function 
$\rho_{xy}$ to quartet networks.

\begin{figure*}\begin{center}
\includegraphics[width=7cm]{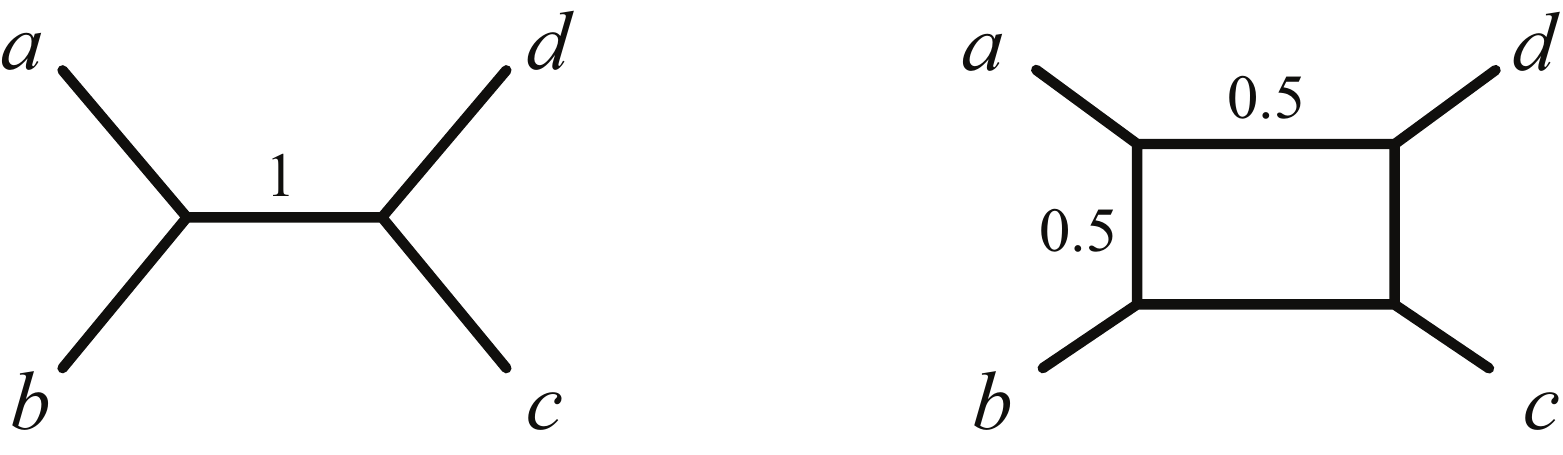} 
\caption{For the tree $Q_{abcd}$ on the left,
$\rho_{ab}( Q_{abcd})=0$ and $\rho_{ac} ( Q_{abcd} ) = 1$, since
$a$ and $c$ are separated by $ab|cd$, but $a$ and $b$ are not.
For the quartet network $Q_{abcd}$ on the right, $\rho_{ab}( Q_{abcd})=1/2$ and
$\rho_{ac}( Q_{abcd})=1$, since the trees displayed by $Q_{abcd}$ are
$ab|cd$ and $ad|bc$.}\label{fig:quartets_wts} 
\end{center}\end{figure*}

\begin{defi} \label{def:rhonet} Let $Q_{xyzw}$ be an unrooted level-1
  4-taxon network on 4 distinct taxa $x,y,z,w\in X$. After contracting
  all 2- and 3-cycles, and suppressing degree 2 nodes, we obtain a
  network $\widetilde Q_{xyzw}$ that is either a tree or has a single
  4-cycle. Let
$$\rho_{xy}( Q_{xyzw})=\begin{cases} 
  0 &\text{ if  $\widetilde Q_{xyzw}$ has form $xy|zw$},\\
  1/2 & \begin{aligned} &\text{if  $\widetilde Q_{xyzw}$ has a 4-cycle}\\ 
&\text{\ \ \ \ \ \ \ \ \ \ \ \ with $x,y$ adjacent, }\end{aligned}\\
  1 & \text{ otherwise}.
 	\end{cases}
$$		
\end{defi}

In the case  $Q_{xyzw}$ is a tree, this definition agrees 
with that in Theorem \ref{prop:treequart}.   An intuitive way of viewing this extension to networks 
is to observe that when $\widetilde Q_{xyzw}$ is a 4-cycle, $\rho_{xy} (Q_{xyzw})$ 
is the average of the values of $\rho_{xy} (T)$ for $T \in \mathcal G (\widetilde Q_{xyzw})$,
so $\rho_{xy}$ measures how separated $x$ and $y$ are on $Q_{xyzw}$.
See Figure \ref{fig:quartets_wts}.

\begin{lemma}\label{lem:quartetweight} For a unrooted level-1 network $N^-$, with $k$ cycles of size $\ge 4$, 
and distinct $x,y,z,w\in X$, let  $Q_{xyzw}$ be the induced unrooted 4-taxon network on $x,y,z,w$. Then
$$\rho_{xy}( Q_{xyzw})=\frac 1{2^k}\sum_{T\in\grove}\rho_{xy}(T_{xyzw}).$$
\end{lemma}
\begin{proof}
If $\rho_{xy}( Q_{xyzw})=0$, then there is no $T\in\grove$ with $T_{xyzw}$ separating $x,y$, so the equation holds.
If $\rho_{xy}( Q_{xyzw})=1/2$, then $\widetilde Q_{xyzw}$ has two hybrid edges, which are induced from hybrid edges 
of $N^-$. Each of these is deleted in exactly half of the $2^k$ trees in $\grove$, so  $2^{k-1}$ of the $T\in \grove$ 
have $T_{xyzw}$ displaying a quartet separating $x,y$, and the equation holds. 
Finally, if $\rho_{xy}( Q_{xyzw})=1$, so $\widetilde Q_{xyzw}$ is either a quartet tree separating 
$x,y$, or has a 4-cycle with $x,y$ opposite in its circular ordering, then for all $T\in \grove$, 
$T_{xyzw}$ will display a quartet separating $x,y$, so the equation holds.
\end{proof}
We now define a distance function in terms of quartet networks displayed on the network. 
\begin{defi}\label{def:quartet distance}
Let $N^-$ be an unrooted level-1 network on $X$. Then the  \emph{quartet distance} $d_{Q,N^-}$ is    
$$d_{Q,N^-}(x,y)=2\sum_{z,w\ne x,y}\rho_{xy}(Q_{xyzw})+2n-4,$$
with $x,y \in X$, distinct from $w,z \in X$.
\end{defi}
Note that if $N^-=T$ is a tree, the definition of $d_{Q,N^-}(x,y)$ agrees with equation \eqref{eq:treedis}.
We now prove the network analog of Theorem \ref{prop:treequart}, showing that the network distance
$d_{\omega_{N^-}}$ can be computed from induced quartet networks.
\begin{theorem}\label{prop:splitquartet}
	Let $N^-$ be an unrooted level-1 network on $X$, with $k$ cycles of size $\ge 4$. Then 
	$$d_{\omega_{N^-}}={2^k}d_{Q,N^-}.$$
\end{theorem}

\begin{proof} Using Lemma \ref{lem:quartet tress}, Theorem
  \ref{prop:treequart}, and Lemma \ref{lem:quartetweight}, for $x\ne
  y\in X$,
\begin{align*} d_{\omega_{N^-}}&(x,y)=\sum_{T\in\grove} d_{w_T}(x,y)\\
&=\sum_{T\in\grove} \left (2 \sum_{\text{$q$ on $T$}}\rho_{xy}(q) +2n-4\right ) \\
&=2\sum_{T\in\grove}  \sum_{\text{$q$ on $T$}} \rho_{xy}(q) \ +2^k(2n-4) \\
&=2\sum_{T\in\grove}  \sum_{z,w\ne x,y} \rho_{xy}(T_{xyzw}) \ +2^k(2n-4) \\
&=2  \sum_{z,w\ne x,y} \sum_{T\in\grove} \rho_{xy}(T_{xyzw}) \ +2^k(2n-4) \\
&=2  \sum_{z,w\ne x,y} 2^k \rho_{xy}(Q_{xyzw}) \ +2^k(2n-4) \\
&={2^k} \, d_{Q,N^-}(x,y).
\end{align*} 
\end{proof}
The import of this theorem is that from the induced quartet networks
on $N^-$ we can compute the distance $d_{Q,N^-}$, which is, up to
scaling, $d_{\omega_{N^-}}$, the distance from a weighted split
system. In contrast, computing $d_{\omega_{N^-}}$ directly from
definition requires knowing $\grove$, the collection of trees on $X$
displayed on $N^-$. This lies at the heart of our algorithm for
network inference under the NMSC, as we can obtain information about
induced quartet networks from biological data relatively easily, using
empirical concordance factors, while information about the trees
displayed on the species network does not seem to be directly
obtainable.

Furthermore, since by Corollary \ref{cor:splitcirc} the underlying
quartet weighted split system is circular, we have the following.
 \begin{coro}\label{cor:DQcirc}
   Let $N^-$ be an unrooted level-1 network. Then the distance
   $d_{Q,N^-}$ arises from a weighted circular split system, with
   support $\s(N^-)$.
 \end{coro}

 Thus given sufficient information on induced quartet networks to
 compute $d_{Q,N^-}$, even approximately as in the presence of error,
 methods for analyzing distances from weighted circular split systems,
 such as the NeighborNet algorithm, can be productively applied, as we
 show in the next section.

\section*{Splits graphs from the network quartet distance}\label{sec:interp} 

The last sections have shown a path toward obtaining, under the NMSC
model, the distance associated to the weighted circular split system
$\omega_{N^-}$. But for this to have value, we need to be able to
extract from this distance information about features of $N^-$. While
there is a well developed theory of splits graphs \cite{Dress1992,
  DressTtheory, Dress2004, HusonRuppScorn}, associated to distances
from such split systems, and splits graphs are networks, one can not
hope that such splits graphs give $N^-$ directly. In particular splits
graphs have no directed edges, and are generally not level-1.

Our goal in this section is thus to investigate the relationship
between a level-1 network and the splits graphs obtainable from the
quartet distance for that network. We develop precise rules by which
one can interpret features in a splits graph for $\omega_{N^-}$ to
obtain much information on the topological features of $N ^-$. While
there is some overlap between the results in this section and those of
\cite{Huson2005}, we give a complete presentation as is necessary for
our more detailed results.

\medskip

The tree edges (i.e., the undirected edges) in a level-1 unrooted network $N^-$ can be classified into two
types, extending Definition \ref{def:network} in this setting.  Specifically, a \emph{cycle edge} in $N^-$ is an 
undirected edge in a cycle, and a \emph{cut edge} is an undirected edge that is not a cycle edge. 
Any $k$-cycle in $N^-$ is then composed of $k-2$ cycle edges and 2 hybrid edges.   

These notions extend to trees displayed on networks.
For any $T\in\grove$, the edges of $T$ arise from those of  $N^-$ in one of the following ways:
\begin{enumerate}
\item[1.] An edge $\bar e$ of $T$ is obtained directly from an edge of
  $N^-$.  Then $\bar e$ is called a cycle or cut edge of $T$ according
  to its classification in $N^-$.
\item[2.] An edge $\bar e$ of $T$ is obtained from several edges of
  $N^-$ by suppressing internal nodes of degree 2.  Since $N^-$ is
  level-1, at least one of these conjoined edges of $N^-$ is a cut
  edge, so we refer to $\bar e$ as a cut edge of $T$.
\end{enumerate}

As we show below, cut edges in $N^-$ correspond to splits $s \in \s(N^-)$ that
occur on every $T \in \grove$, while a split $\bar s$ derived from a cycle edge on $T$ does
not occur on every $T' \in \grove$.  
Moreover, we see that edges in 2-cycles and 3-cycles on $N^-$ induce \emph{only} cut edges on 
any $T\in \grove$. For $k\ge 4$, a $k$-cycle on $N^-$ will induce $k-3$ cycle edges on any $T\in \grove$, 
since one hybrid edge is deleted, one hybrid edge is conjoined with its descendent  cut edge, and 
one cycle edge is conjoined with a cut edge.

A split $s\in \s(N^-)$ is called a \emph{cycle split} (respectively, a
\emph{cut split}) if $s=s_{\bar e}$ for a cycle edge (respectively, a
cut edge) $\bar e$ on some $T\in\grove$. Note that the cut splits are
precisely those splits obtained from $N^-$ by deletion of a cut edge,
and that these two classes of splits form a partition of $\s(N^-)$.

\smallskip

In the next lemma, we prove that the quartet weight function $\omega_{N^-}$ on an unrooted network $N^-$ 
carries no information about 2- or 3-cycles.

\begin{lemma} \label{lem:23cycle} Let $N^-_c$ be the graph obtained
  from a level-1 binary network $N^-$ by contracting each 2- and
  3-cycle to a vertex and then suppressing degree 2 nodes. Then
  $\omega_{N^-_c}=\omega_{N^-}$.
\end{lemma}
\begin{proof}
If one or the other hybrid edge in a 2- or 3-cycle on $N^-$ is deleted, the resulting network has the same 
topology as obtained by contracting the cycle. Thus $N^-$ and  $N^-_c$ display the same topological trees. 
\end{proof}
  
In the next lemma, we formalize some observations made above.
 
\begin{lemma}\label{lem:treesplit}
Let $s\in\s(N^-)$ for a level-1 binary network $N^-$. Then the following are equivalent:
\begin{enumerate}
\item[(1)] \label{item:all}$s\in \s(T)$ for all $T\in\grove$,
\item[(2)] \label{item:one} On every $T\in \grove$ there is a cut edge $\bar e$ such that $s=s_{\bar e}$,
\item[(3)] \label{item:comp} $s$ is compatible with every $s'\in \s(N^-)$.
\end{enumerate}
\end{lemma}
\begin{proof} Clearly  (2) implies (1).
To see that (1) implies (2), suppose on some tree $T\in \grove$ there is a cycle 
edge $\bar e$ with $s=s_{\bar e}$. Then $\bar e$ arises from a cycle edge in $N^-$ and that cycle has hybrid edges 
$e_1$ and $e_2$, where $e_1$ was deleted to form $T$. Then no tree $T'\in \grove$ which is formed 
by deleting $e_2$ will display $s$. This contradicts (1).

That (1) implies (3) is immediate. For the converse, observe that
since $N^-$ is binary, each $T\in\grove$ is binary. But the set of
splits on a binary tree is maximal with respect to compatibility, so
(3) implies (1).
\end{proof}

The equivalences in Lemma \ref{lem:treesplit} imply that a split from a cycle edge in some $T\in\grove$ is 
incompatible with some split from a cycle edge on some other tree in $\grove$, an observation we 
further refine in the following lemma.

\begin{lemma} \label{lem:cyclesplit}
Let $s,s'\in \s(N^{-})$ for a level-1 binary network $N^-$. Then $s,s'$ are incompatible if, and only if, 
there are cycle edges $e,e'$ (not necessarily distinct) on $N^{-}$  in the same cycle $C$, and $T,T'\in \grove$ 
such that $e,e'$ induce cycle edges $\bar e, \bar e'$ on $T,T'$ with $s=s_{\bar e}, s'=s_{\bar e'}$ and $T,T'$ 
were obtained by deleting different hybrid edges from $C$.
\end{lemma}

\begin{proof} Consider incompatible $s,s'\in \s(N^{-})$. Then by Lemma \ref{lem:treesplit}, there exist 
$T,T'\in \grove$ with cycle edges $\bar e, \bar e'$ where $s=s_{\bar e}, s'=s_{\bar e'}$. The edges 
$\bar e, \bar e'$ are induced from cycle edges $e, e'$ in $N^-$. 

Suppose $e,e'$ are in cycles $C\ne C'$. Now $T$ determines a hybrid edge of $C$ whose 
removal from $N^-$, along with the removal of $e$, determines the split $s$, and $T'$ similarly 
determines a hybrid edge of $C'$. Removing these two hybrid edges, together with one hybrid 
edge from every other cycle on $N^-$ determines a tree $T''\in \grove$. But $T''$ has both $s,s'$ 
as displayed splits, which implies they are compatible. Thus $e,e'$ must be in the same cycle on $N^-$.

Moreover, $T,T'$ must be obtained by deleting different hybrid edges in the cycle containing $e,e'$, 
since if the same hybrid edge were deleted, the splits $s,s'$ would again be displayed on a 
common tree, and hence be compatible.

\smallskip

For the converse, suppose $e,e'$ are cycle edges in cycle $C$ of $N^-$, which induce cycle edges in trees 
$T,T'\in \grove$, where $T,T'$ are obtained by deleting different hybrid edges in $C$. Let 
$X=X_0 \sqcup X_1\sqcup X_2 \sqcup\dots \sqcup X_{m-1}$ be the partition of $X$ obtained from the 
connected components of the graph resulting from removing all edges of $C$ from $N^-$. Suppose 
further that the ordering of these sets reflects the ordering around the cycle, so that $X_0$ is 
descendants of the the hybrid node, and $X_1,X_{m-1}$ are its neighbors, etc. Then, without loss of generality, 
we may assume that split $s_e$ displayed on $T$ is $X_0\cup\dots \cup X_k| X_{k+1}\cup\dots \cup X_{m-1}$ 
with $1\le k\le m-3$, while the split $s_{e'}$ displayed on 
$T'$ is $X_0\cup X_{m-1}\cup\dots \cup X_{\ell+1}| X_{\ell}\cup\dots\cup X_1$ with  $2\le \ell\le m-2$. 
These splits are incompatible as claimed.
\end{proof}

Split networks \cite{HusonRuppScorn}, also known as splits graphs,
provide a valuable visual tool for interpreting split systems.  In
what follows, we use the terminology `splits graph' exclusively to
avoid confusion with the species networks $N^+$ and $N^-$ associated
with the NMSC.

In a splits graph, each edge is colored by exactly one of the splits,
with each split possibly coloring multiple edges. Deleting all edges
with a common color leaves two connected components, with taxon labels
on the components giving the split sets.  Unfortunately splits graphs
are generally not uniquely determined by split systems. However, since
the split systems of interest here arise from level-1 networks $N^-$,
and thus are circular by Corollary \ref{cor:splitcirc}, we can impose
an additional requirement, that of `frontier-minimality' developed
below, to determine most features of $N^-$ from interpretation of a
frontier-minimal splits graph.  The \CNA\ of \cite{Dress2004} is the
key to both showing split graphs with this additional property exist
in this case, and producing them in specific instances.

Recall that the \emph{frontier} of a planar graph is the subset of
edges adjacent to the unbounded component of its complement in the
plane (more informally, the ``outside" edges of the graph). A graph is
\emph{outer-labelled} if the labelled vertices are in the frontier.
Also, a \emph{blob} on a network is a maximal set of edges in
undirected edge-intersecting cycles.  On an unrooted level-1 network
such as $N^-$, a blob is simply an undirected version of a cycle.

\begin{lemma} \label{lem:goodsplitgraph} Let $S=S_c \sqcup S_i$ be a
  circular split system, with $S_c$ the subset of splits compatible
  with all others in $S$, and $S_i$ those incompatible with at least
  one other. Then the \CNA\ of \cite{Dress2004} produces an
  outer-labelled planar splits graph $N_S$ such that
\begin{enumerate}
\item[1.] If $s\in S_c$, then $s$ colors exactly one edge in the
  frontier of $N_S$, and this edge is not in any blob.
\item[2.] If $s\in S_i$, then $s$ colors precisely 2 edges in the
  frontier (and possibly additional edges not in the frontier) which
  lie in the same blob.
\item[3.] If $s,s' \in S_i$ are incompatible, then they color frontier edges in the same blob.
\end{enumerate}
\end{lemma}
\begin{proof}
The \CNA\  works iteratively, by adding new vertices and edges as each split is considered in some order, to produce 
an outer-labelled splits graph \cite{Dress2004}.

We may assume the trivial splits are in the system. The algorithm
begins with these splits represented by a star tree, and the stated
properties hold.  Each time an additional split $s$ is considered, the
algorithm first determines if this split is incompatible with the
current graph $G_i$.  If it is, the algorithm `duplicates' parts of
the frontier, composed of some edges labelled by splits incompatible
with $s$, joining the duplicated section to the old part by `ladder'
edges colored by the new split $s$ to form $G_{i+1}$.  This makes the
frontier grow by 2 edges colored by $s$, and ensures that any splits
incompatible with $s$ previously coloring only one frontier edge in
$G_i$, now color two frontier edges in $G_{i+1}$.  Then any two edges
colored by the same split lie in the same blob, as do frontier edges
coloring incompatible splits.

If the new split $s \in S_c$, then, reminiscent of the tree-popping algorithm, 
a single new edge in $G_i$ is introduced to form $G_{i+1}$ and is colored by $s$.
This new edge is not in a blob.
\end{proof}

This coloring of edges in the frontier of the splits graph produced by the \CNA\  
can be characterized in an alternative, less algorithmic, way.

\begin{defi}  If $S$ is a circular split system on $X$, then an outer-labelled planar splits graph $N_S$
on $S$ is \emph{frontier-minimal}, if $N_S$ contains the minimal number of frontier edges among
all outer-labelled planar splits graphs on $S$.
\end{defi}

\begin{prop} \label{prop:minsplitgraph} Any frontier-minimal splits
  graph $N_S$ for a circular split system $S$ has properties (1), (2),
  and (3) of Lemma \ref{lem:goodsplitgraph}.  Moreover, the \CNA\
  produces a frontier-minimal splits graph.
\end{prop}
\begin{proof} First, observe that each split in $S$ must label at least one frontier edge, else 
deletion of edges labelled by that split would not disconnect $N_S$.

Next, recall that the operation of contraction of a split $s$ in a splits
graph for $\s$, which identifies the two endpoints of each edge
labelled by $s$ and deletes the edge, yields a splits graph for
$\s\smallsetminus \{s\}$ (Lemma 5.10.1 of \cite{HusonRuppScorn}).
Moreover, frontier edges resulting from contraction must arise from
frontier edges in the original splits graph. If $s, s'\in S_i$ are
incompatible splits in a splits graph for $\s$, then by contracting
all other splits we obtain a split network depicting only these
two. Now if it were the case that only one frontier edge in this
splits graph were labelled by $s$, deletion of that edge must separate
the graph. But then, since $s'$ is incompatible with $s$, $s'$ must
label edges whose deletion disconnects each of the components obtained
by deleting the $s$ edge. But this implies that deleting only the $s'$
edges in $N_S$ separates the graph into at least 3 components, which
contradicts that it is a splits graph.  Thus $s$ labels at least 2
frontier edges.

It follows that any splits graph has at least $|S_c| +2|S_i|$ frontier
edges, and since this minimal count is achieved by the splits graph
output from the \CNA, a frontier-minimal splits graph has $|S_c|
+2|S_i|$ frontier edges.

Furthermore, in any splits graph for $S$ each element of $S_i$ colors
at least two frontier edges and each element of $S_c$ at least one.
It then follows from the count of frontier edges in a frontier-minimal
splits graph that the elements of $S_i$ color precisely two frontier
edges, and elements of $S_c$ precisely one.  The single frontier edge
labelled by an element of $S_c$ cannot lie in a blob, since otherwise
deleting it would not disconnect the graph.  This establishes
properties (1) and (2) of Lemma \ref{lem:goodsplitgraph}.

Finally, if $s\in S_i$, then for any $s'\in S_i$ incompatible with
$s$, contracting all splits but $s,s'$ in a frontier-minimal splits
graph must give a splits graph with four frontier edges. By
considering all possible such graphs, these edges must form a 4-cycle
with edges labelled in order $s,s',s,s'$. Since these four edges are
in the same blob on this graph, they must be in the same blob in the
original graph.
\end{proof}

In \cite{Dress2004} it is shown that the \CNA\ produces a splits graph
minimal in a different sense: It has the smallest number of edges
among all splits graphs whose bounded faces are parallelograms (i.e.,
quadrilaterals with opposite sides sharing colors).  This addresses
internal structure of the blobs, which our notion of frontier-minimal
ignores. We have not investigated whether the two notions of
minimality are equivalent, nor to what extent a frontier-minimal
splits graph for a circular split system is unique.

\smallskip

The \emph{tree of blobs} of a graph is the graph obtained by
contracting edges and vertices in each blob to a single vertex.

\begin{coro}\label{lem:treeofblobs} 
  The tree of blobs of a level-1 network $N^-$ is isomorphic to the
  tree of blobs of a frontier-minimal splits graph for
  $\s(N^-)$. \end{coro}
\begin{proof}
  The tree of blobs of $ N^-$ displays precisely those splits
  associated to cut edges of $N^-$. By Lemma \ref{lem:treesplit},
  these are precisely the splits compatible with all others in
  $\s(N^-)$, and by Proposition \ref{prop:minsplitgraph}, the tree of
  blobs of a frontier-minimal splits graph displays the same set.
\end{proof}

To go further, we investigate how the structure of a blob (a cycle) in
$N^-$ corresponds to a related structure of a blob (\emph{not}
generally a cycle) in a frontier-minimal splits graph for $\s(N^-)$.
The following, which characterizes splits associated to a cycle in
$N^-$, follows straightforwardly from definitions, so a formal proof
is omitted. The argument is readily supplied by considering Figure
\ref{fig:cycle}, which depicts a single cycle in $N^-$, and the two
networks obtained from it by deleting one or the other hybrid edge.
 \begin{figure*}\begin{center}
		\includegraphics[scale=.6]{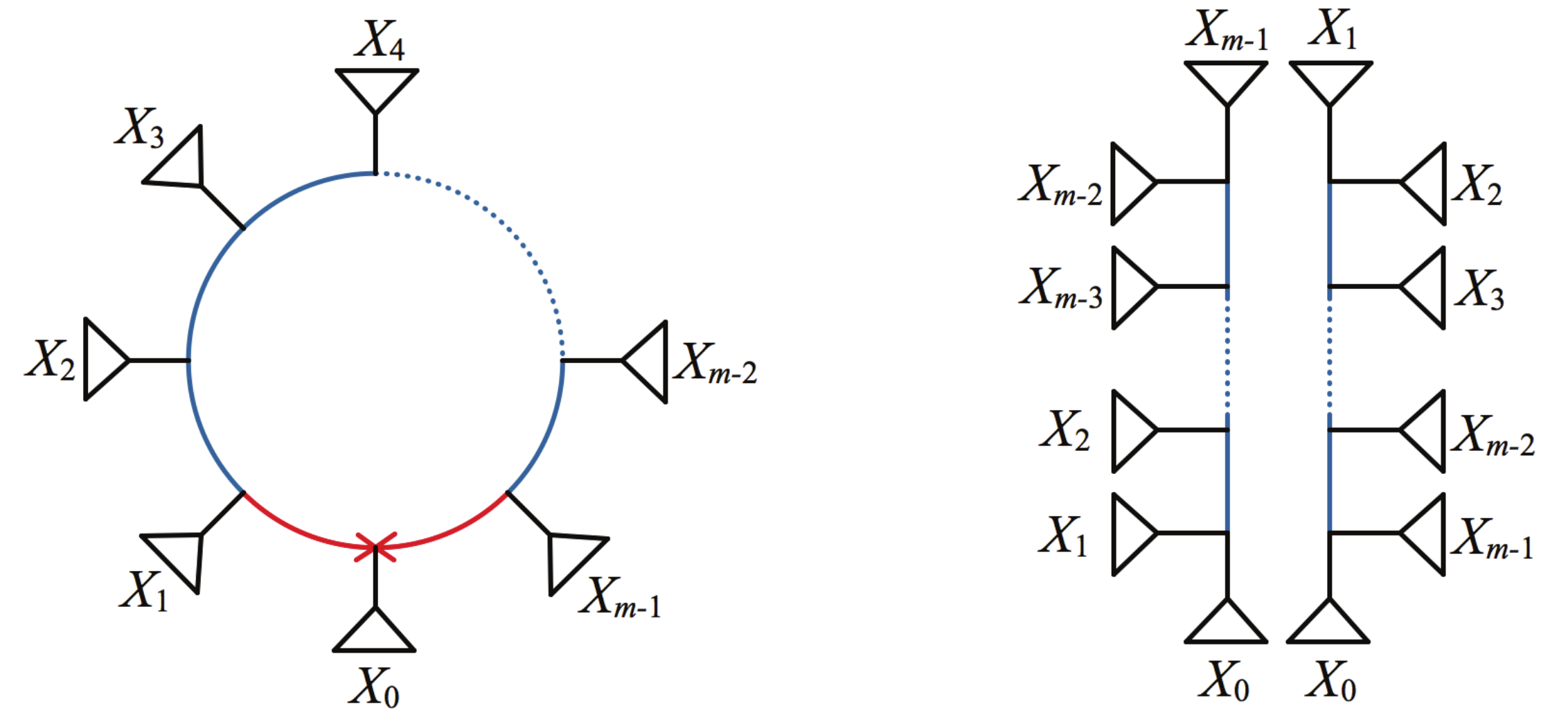}
		\caption{ (L) A cycle in a level-1 network $N^-$, and (R)
                  the two simpler networks produced from it by
                  deleting one hybrid edge.  The cycle edges in these
                  networks that arise from the original cycle are
                  shown in blue.  If $N^-$ has a single cycle, then
                  the networks on the right are the two trees in
                  $\grove$.  }\label{fig:cycle}
\end{center}\end{figure*}

\begin{lemma} \label{lem:cyclesplits} Suppose a level-1 unrooted
  network $N^-$ has $k$ cycles of size $\ge 4$.  Let $C$ be an
  $m$-cycle on $N^-$, $m\ge 4$,and $X=X_0 \sqcup X_1\sqcup X_2
  \sqcup\dots \sqcup X_{m-1}$ the partition of $X$ obtained from the
  connected components of the graph resulting from removing all edges
  of $C$ from $N^-$. Suppose further that the ordering of these sets
  reflects the ordering around the cycle, so that $X_0$ is the
  descendants of the the hybrid node, and $X_1,X_{m-1}$ are its
  neighbors, etc. (see Figure \ref{fig:cycle}).  Then the cycle splits
  in $\s(N^-)$ arising from edges in $C$ are
\begin{align}
X_0\cup X_1\cup\dots \cup X_i| X_{i+1}\cup\dots &\cup X_{m-1},\label{eq:s1}\\ & 1\le i\le m-3,\notag\\ 
X_0\cup X_{m-1} \cup \dots \cup X_{j+1}| X_{j}\cup\dots &\cup X_1,\label{eq:s2}\\ &2\le  j\le m-2, \notag
\end{align}
all with $\omega_{N^-}(s)= 2^{k-1}$. Those splits of the form
\eqref{eq:s1} (respectively \eqref{eq:s2}) are compatible with all
others of that form.  Splits of the form \eqref{eq:s1} are
incompatible with those of the form \eqref{eq:s2}.  Splits of the form
\eqref{eq:s1} or \eqref{eq:s2} are compatible with all other elements
of $\s(N^-)$.

Moreover, $(X_0,X_1,X_2,\dots, X_{m-1})$ is the only circular ordering
of the $X_i$ consistent with these splits, and with $X_m=X_0$ the
number of cycle splits arising from $C$ that separate $X_i$ from
$X_{i+1}$ is
$$\begin{cases} 
m-3&\text{ if $i=0,m-1,$}\\
1&\text{ if $i=1,m-2,$}\\
2&\text{ otherwise.}
\end{cases}$$
\end{lemma}

The next lemma describes the part of the frontier in a
frontier-minimal splits graph arising from splits associated to a
single $m$-cycle, a description which will be used later to identify
hybrid edges.
\begin{lemma}\label{lem:4cyclemdart} With notation as in Lemma \ref{lem:cyclesplits},
  a frontier-minimal splits graph for the cycle splits $\s(C)$ arising
  from a single cycle $C$ of size $m\ge 4$ in $N^-$ forms a single
  blob whose frontier is a cycle of size $4(m-3)$. Moreover, there are
  distinct vertices labelled in circular order by $X_0,X_1,\dots,
  X_{m-1}$ along the frontier, with the number of edges between labels
  $X_i, X_{i+1}$ equal to the number of splits in $S(C)$ that separate
  $X_i,X_{i+1}$.
\end{lemma}

\begin{proof} Consider two splits associated to the cycle. By Lemma
  \ref{lem:cyclesplits}, they are either incompatible, or they are
  both incompatible with a third split from the same cycle.  By Lemma
  \ref{lem:goodsplitgraph}, they therefore color edges in the same
  blob, and it follows that there is only one blob in the splits
  graph. Since by Lemma \ref{lem:cyclesplits} there are $2(m-3)$
  splits associated to the cycle, by Proposition
  \ref{prop:minsplitgraph} the blob has $|S_c| +2|S_i| = 4(m-3)$ edges
  in its frontier.

Also by Lemma \ref{lem:cyclesplits} there exist splits separating any $X_i,X_j$, $i\ne j$, so the 
$X_i$ must label distinct vertices in the frontier. Since any split separating $X_i$ and $X_{i+1}$ labels at 
least one edge in any frontier path between them, the number of edges in a minimal frontier 
path between $X_i$ and $X_{i+1}$ is at least the number of splits separating them. This then 
implies that the $X_i$ must be in order along the frontier, at the distances claimed.
\end{proof}

 \begin{figure*}\begin{center}
\includegraphics[scale=.3]{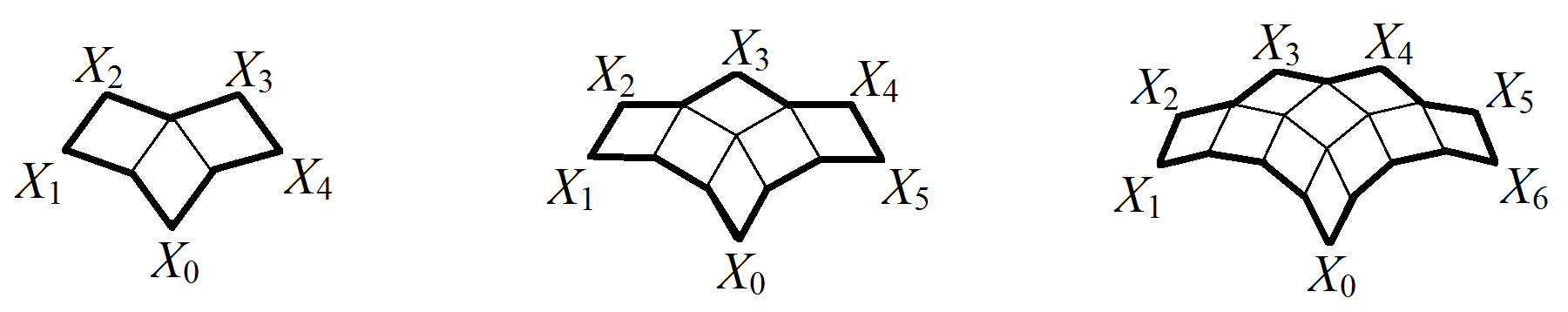}
\caption{An $m$-dart, for $m=5,6,7$ respectively. The frontier edges,
  shown in bold outline, are characterized in the text.  The outer
  vertices labelled by the $X_i$ are the corners. The point of the
  dart is the unique corner which is $m-3$ frontier edges away from
 the closest corners. }\label{fig:darts}
 \end{center}\end{figure*}
 
Now suppose $C$ is an $m$-cycle in $N^-$. If $m=4$, this lemma indicates that a frontier-minimal 
splits graph for the splits associated to $C$ is also a 4-cycle, that is, the undirected 
version of the cycle. However, if $m\ge 5$, the splits graph is more complicated, having frontier as those depicted
in the examples of Figure \ref{fig:darts}. We refer to such blobs as \emph{$m$-darts}.
The \emph{corners} of the $m$-dart are the vertices on the frontier of the dart that are 
labeled by sets of taxa $X_i$.  The \emph{point} of the $m$-dart, labelled by $X_0$, is the unique corner
that is $m-3$ frontier edges away from its two closest corners. Thus in a 
closed walk around the frontier of the dart starting at the point, the number of edges between consecutive corners is
$$m-3,1,2,2,\dots,2,2,1, m-3.$$

Putting all this together, we have the following.

\begin{theorem}\label{thm:splitsgraphinterp}
Given a level-1 unrooted network $N^-$, the frontier of any frontier-minimal splits graph for $S(N^-)$
is the graph obtained from $N^-$ by the following steps:
\begin{enumerate}
\item[1.] Contract any 2- and 3-cycles,
\item[2.] Undirect the hybrid edges in any 4-cycles,
\item[3.] Replace any $m$-cycle, $m\ge 5$, with the frontier of an $m$-dart so that the point is 
at the hybrid node and with the $m$ cut edges incident to the cycle connected to the corners of the 
dart in the same circular ordering as in the cycle. 
\end{enumerate}
\end{theorem}

\begin{proof} 
  By Lemma \ref{lem:23cycle}, we may assume $N^-$ has no 2- or
  3-cycles. Let $k$ denote the number of cycles of size $\ge 4$ on
  $N^-$, and $G$ a frontier-minimal splits graph for $\s(N^-)$.

By Corollary \ref{lem:treeofblobs}, the tree of blobs of $N^-$ and the tree of blobs of $G$ are isomorphic, so 
we identify them. Moreover, since cycles in $N^-$ are vertex-disjoint, each cycle of size $m\ge 4$  on $N^-$ 
gives rise to a node of degree $m$ in the tree of blobs, so the tree of blobs has $k$ multifurcations. This 
implies $G$ has at least $k$ blobs.   \emph{A priori} it is possible that $G$ has more than $k$ blobs, 
since if two blobs in $G$ shared a vertex they would be collapsed to a single node in the tree of blobs.

By Proposition \ref{prop:minsplitgraph} property (3), frontier edges
of G colored by splits associated with a single cycle of $N^-$ all lie
in a single blob of $G$, since Lemma \ref{lem:cyclesplits} shows two
such cycle splits are either incompatible, or both incompatible with a
third.  Moreover, since the tree of blobs of $N^-$ (and $G$) has
exactly $k$ vertices corresponding to cycles in $N^-$, it follows that
$G$ has exactly $k$ blobs, which are vertex disjoint, and each blob
has only splits associated to a single cycle of $N^-$ coloring its
frontier edges.  This establishes a one-to-one correspondence between
cycles in $N^-$ and blobs in $G$, according to the coloring of
frontier edges

Fixing a cycle $C$ on $N^-$, and contracting all edges of $G$ not
labeled by splits associated to $C$ preserves the frontier of the blob
of $G$ corresponding to $C$.  By Lemma \ref{lem:4cyclemdart}, this
frontier is either a 4-cycle (if $m=4$) or an $m$-dart (if $m\ge 5$).
Moreover, the partition of $X$ according to the connected components
of $N^-$ with $C$ deleted is the same as that from the labeled corners
of the 4-cycle or $m$-dart, with the same circular ordering, and in
the case $m\ge 5$ the descendants of the hybrid node of $C$ label the
dart's point. Thus both $C$ in $N^-$ and the blob of $G$ associated to
$C$ must map to the same multifurcation in the tree of blobs, and the
frontier of $G$ must have the form described.
\end{proof}

\begin{figure*}
	\includegraphics[width=.92\textwidth]{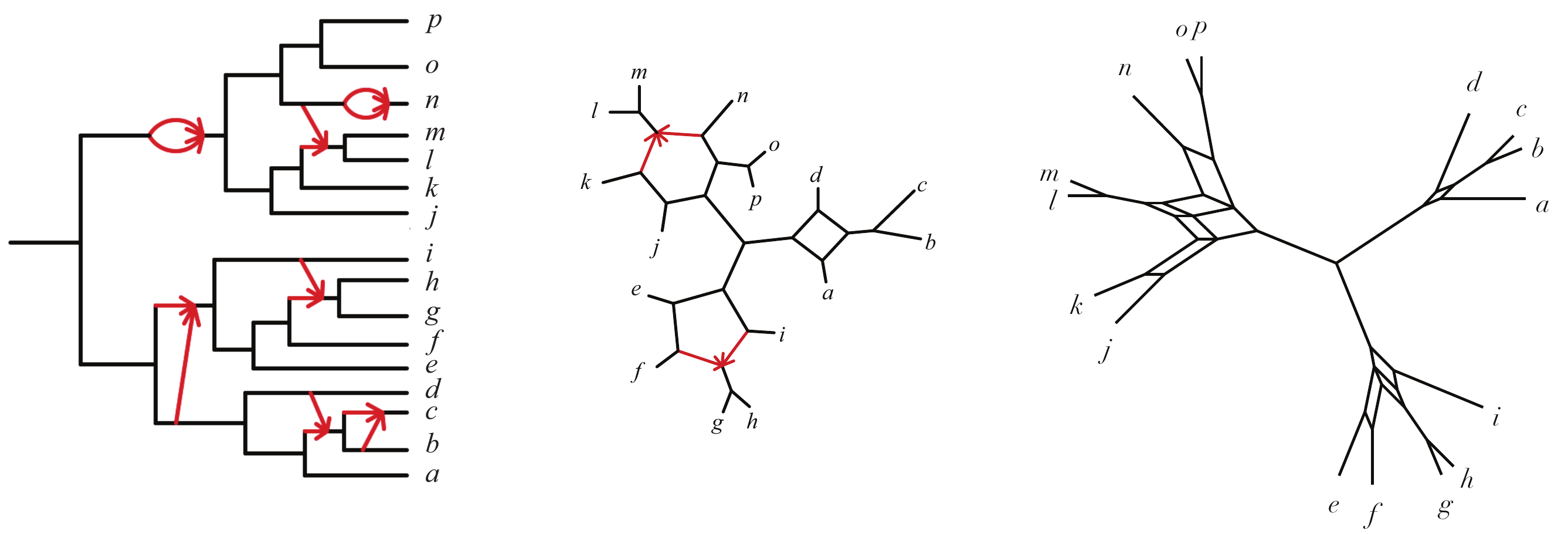}
         \caption{(L) A rooted level-1 network $N^-$, (C) the unrooted
           network obtained from it by contracting 2- and 3-cycles and
           undirecting 4-cycles, and (R) a frontier-minimal splits
           graph corresponding to it by Theorem
           \ref{thm:splitsgraphinterp}.  Note the splits graph has a
           4-cycle, a 5-dart, and a 6-dart, arising from the 4-, 5-,
           and 6-cycles of $N^-$. The metric structure of the splits
           graph, which is not described by Theorem
           \ref{thm:splitsgraphinterp}, reflects the split weights as
           defined by Definition \ref{def:netsplitwts}.
         }\label{fig:nanuq}
\end{figure*}

Figure \ref{fig:nanuq} illustrates this theorem for a particular
network $N^-$.  Note that the theorem only describes the topological
structure of the splits graph. The metric splits graph's structure
depends on details of the network beyond the analysis of the theorem,
as is seen in Definition \ref{def:netsplitwts} of the split weights.

\medskip

Importantly for applications, one can apply Theorem
\ref{thm:splitsgraphinterp} ``in reverse" to obtain information about
the network $N^-$ from the frontier-minimal splits graph for
$\s(N^-)$.  Indeed, although the correspondence between level-1
networks $N^-$ and frontier-minimal splits graphs as described in
Theorem \ref{thm:splitsgraphinterp} is not one-to-one, the only
information lost from $N^-$ is that of the existence of 2- and
3-cycles and the determination of the hybrid node in a 4-cycle.  The
specific geometry of the frontier of an $m$-dart in $G$ for $m\ge 5$
allows one to identify such $m$-cycles and hybrid nodes in $N^-$.  In
conjunction with previous sections of this paper, this recovers the
main result of \cite{Banos2018}:

\begin{coro}\label{cor:netid} Under the NMSC model on a level-1
  network $N^+$, for generic parameters, the network obtained from
  $N^-$ by suppressing 2- and 3-cycles and undirecting 4-cycles is
  identifiable.
\end{coro}

\medskip

Beyond providing a different argument for this corollary, Theorem
\ref{thm:splitsgraphinterp} provides theoretical underpinnings to a
practical algorithm for (partial) network topology inference from a
sample of gene trees, as outlined in the next section.


\section*{The NANUQ algorithm for inference of phylogenetic networks}\label{sec:nanuq}
Here we revisit and formalize the NANUQ algorithm sketched in the introduction.

\begin{alg}[NANUQ] 
\

Input: A collection of unrooted topological gene trees on subsets of
a taxon set $X$, such that each 4-element subset of $X$ appears on at
least one tree; and two hypothesis testing levels $0<\alpha,\beta<1$.
\be
\item[1.]  \label{item:count} For each subset of 4 taxa, determine the
  empirical quartet counts across the gene trees for each of the 3
  resolved topologies. If all four taxa are not on a gene tree, that
  tree does not contribute to the counts. These 3 counts form an
  empirical \emph{quartet count concordance factor} (\qc) vector for
  the 4 taxa.
\item[2.] \label{item:testing} For each set of 4 taxa, apply two
  statistical hypothesis tests to its \qc, with levels $\alpha,\beta$,
  as described below, to determine whether to view the \qc as
  supporting (1) a star tree, (2) a resolved tree, or (3) a 4-cycle
  network on the taxa. In cases (2) and (3), use the maximum
  likelihood estimate of the topology from the \qc to determine which
  tree or network is supported.
\item[3.]  \label{item:dist} Use the quartet networks/trees from the
  previous step to construct a network quartet distance between taxa,
  as in Definition \ref{def:quartet distance}, with the modification
  described below for unresolved quartets.
\item[4.] \label{item:NN} Use the NeighborNet Algorithm
  \cite{Bryant2004} to determine a weighted circular split system
  approximating the quartet distance.
\item[5.]  \label{item:CNA} Use the \CNA\ \cite{Dress2004} to
  determine a frontier minimal splits graph for the circular system.
  \ee Output: A splits graph to interpret via Theorem
  \ref{thm:splitsgraphinterp} for features of $N^+$.
\end{alg}

To analyze the running time for this algorithm, suppose $|X|=n$ and
the input set contains $m$ trees. First note that tallying displayed
quartets in Step 1 can be done in time $\mathcal O (n^4m)$, as
discussed in \cite{Rhodes2019}.  The hypothesis tests for Step 2 are
performed in constant time for each set of 4 taxa, for a total of
$\mathcal O (n^4)$. Step 3 in which the distance is computed requires
running through the inferred quartet trees and networks for an
additional time of $\mathcal O (n^4)$. The NeighborNet algorithm in
Step 4 takes time $\mathcal O (n^3)$ \cite{Bryant2004}.  Since
NeighborNet can produce positive weights for all $\mathcal O (n^2)$
splits consistent with some circular ordering of the taxa, results
from \cite{Dress2004} show that the time for the \CNA\ in Step 5 is
$\mathcal O (n^4)$.  Thus the total time for NANUQ is $\mathcal O
(n^4m)$.

We implemented Steps 1, 2, and 3 of the NANUQ algorithm in an R
package \texttt{MSCquartets}, with a function accepting an input file
of (metric or topological) Newick gene trees, and producing an output
file of the network quartet distances computed from this data.  When
this file is opened by SplitsTree4 \cite{Huson2006}, Steps 4 and 5 are
performed. With these implementations, we have found Step 1 by far
dominates computational time, as is consistent with the running time
analysis. However, the use of R probably slows computations
considerably over what could be achieved.

The package \texttt{MSCQuartets} is currently available on request from the authors, and will be made 
publicly downloadable after further refinement.

\subsection*{Testing Empirical Quartet Counts} \label{ssec:test} 

The statistical tests in Step 2 of the NANUQ algorithm, based on \cite{Mitchell2018}, require further
explanation.

\smallskip

We use a hypothesis testing framework, in which two tests are performed. 
One test is used to decide whether the topological signal in a \qc is strong enough to justify belief in 
any resolved network or tree, as opposed to viewing the quartet as unresolved.
The second test is used to decide if the \qc supports a 4-cycle network or a tree. The particular network or tree is 
then chosen via maximum likelihood. 

These tests are performed for each set of four taxa, as if all quartet
gene trees are independent.  Of course, these are not independent,
since the quartet trees are subtrees of the same gene trees, and under
the NMSC these gene trees are assumed to have formed on the same
species network.  Since the lack of independence depends in part upon
the species network parameter, which is unknown and sought, it is not
clear how one might compensate for it.  However, treating summary
statistics as independent when they are not also underlies
phylogenetic inference schemes built on \emph{pseudo}-likelihood
(e.g., SNaQ) and seems a necessary and acceptable concession for
developing fast and tractable methods.

\medskip

Suppose for a set of 4 taxa, one has tabulated the counts of the
quartets displayed on gene trees in a sample, obtaining the \qc.
Under the NMSC model, these counts can be viewed as a multinomial
sample from the distribution determined by the theoretical CF.
Normalizing by the total count, we obtain an empirical CF which
estimates the theoretical one.  Because this empirical CF is computed
from a finite sample, it is unlikely that it lies exactly where the
theoretical CF would as shown in Figure \ref{fig:simplex}.  However,
an appropriate statistical test can be used for deciding whether the
\qc supports a quartet tree or network under the NMSC.

\smallskip

Specifically, for a fixed \qc we first perform a hypothesis test for a
star tree. More formally, under the NMSC the null hypothesis
is $$H_0\tc \text{ The \qc arises from a 4-taxon star tree.}$$ The
alternative hypothesis is that the \qc may have arisen from either a
resolved tree or a network under the NMSC, or that the NMSC model
somehow does not apply.  The NANUQ algorithm focuses exclusively on
the first interpretation of the alternative, assuming that all data
arises from the NMSC.

As the star tree has theoretical CF $(1/3,1/3,1/3)$, we perform this
test by computing the likelihood ratio statistic from the three
quartet counts in \qc, using a $\chi^2$ distribution with 2 degrees of
freedom to compute a $p$-value.  With level $\beta$ chosen for the
test, we reject the star tree hypothesis for $p$-values smaller than
$\beta$.  (Note that $\beta$ is used here as the size of the rejection
region for the test, \emph{not} the probability of a type II error.)
For larger $p$-values, we fail to reject the star tree.

As will be shown in Theorem \ref{thm:statcons} below, under the NMSC
on a binary level-1 network for any level $\beta > 0$, the probability
that this test always rejects quartet star trees, approaches 1 as the
sample size (number of gene trees) goes to infinity.  Nonetheless,
with finite and noisy data (perhaps due to gene tree inference error),
this test is important to prevent interpreting a \qc that is nearly
uniform from indicating support for a particular tree or network
topology.  Performing this test allows for the suppression of weak and
possibly erroneous signals in data sets of finite size.
 
 \smallskip
 
 The second hypothesis test is to assess support for a tree-like
 quartet vs. a 4-cycle.  Under the NMSC, we formulate a null
 hypothesis of
$$H_0\tc \text{ The \qc is tree-like},$$
with alternative that \qc is not tree-like.  
Since underlying the NANUQ algorithm is the assumption that gene tree data arose from
the NMSC, rejecting the null hypothesis is interpreted
as giving evidence that
the quartet network has a 4-cycle.  That is, rejecting the null hypothesis is interpreted
by NANUQ as support for a $4$-cycle quartet network, ignoring the (measure
0) region where non-tree-like CFs from $3_2$-cycles may coincide with
$4$-cycle CFs.

Geometrically, the model for this null hypothesis is the 3 line
segments in the simplex of Figure \ref{fig:simplex} (L), with the
alternative model the complement of the 3 line segments as shown in
Figure \ref{fig:simplex} (R).  For the test, we compute the likelihood
ratio statistic for these hypotheses.  Using a $\chi^2$ distribution
with 1 degree of freedom (the asymptotic distribution for a resolved
tree) would be a standard approach to obtain a $p$-value for the
statistic.  However, the model space for $H_0$ has a singularity at
the center of the simplex, and justification for the $\chi^2$ depends
on the model being approximated well by its tangent line. As this
approximation fails at the singularity, using a $\chi^2$ approximation
in the vicinity of the singularity may result in poor testing, which
is in this case is quite conservative.  Although the neighborhood of
the singularity on which the $\chi^2$ behaves poorly shrinks as the
sample size $m$ grows, this `bad' neighborhood is present for any
finite sample size.  However, this particular model and its special
geometry at the singularity has been studied extensively in
\cite{Mitchell2018}, where an alternative approximate distribution has
been developed.  We adopt the techniques of that work for use with the
likelihood ratio statistic, to compute $p$-values.

For the NANUQ algorithm with level $\alpha$ for this test, we
interpret a $p$-value greater than $\alpha$ as support for a tree,
with the particular tree topology chosen as the maximum likelihood
estimate from the \qc.  The MLE quartet tree topology is simply the
quartet topology with the largest count in the \qc.  A $p$-value less
than $\alpha$ is interpreted as support for a 4-cycle network, where
the particular $4$-cycle topology supported is the maximum likelihood
estimate from the \qc.  This is determined by which of the 3
triangular regions in the simplex the normalized \qc lies, as in
Figure \ref{fig:simplex} (R).

\smallskip

With two tests being performed in this way, it is possible that for a
particular set of 4 taxa we find that we fail to reject the first
hypothesis (that the \qc arises a star tree) but reject the second
(that it arises from a tree). This can be forced to occur by taking
$\beta$ quite small while $\alpha$ is large, but it may occur for less
extreme values.  In such a situation one must give priority to one
test over the other. We choose to prioritize the first test, so that
in this case we view the tests as supporting a star tree, on the
principle that evidence for hybridization should be judged by the
strictest standards.

\smallskip

The output of NANUQ depends on the choices of significance levels
$\alpha$ and $\beta$, with smaller values of $\alpha$ requiring
stronger evidence for 4-cycles, and smaller values of $\beta$
requiring stronger evidence for any resolution of the $4$-taxon
network.  We view this feature positively, as it requires that users
of NANUQ examine their data and consider the impact of choosing
different levels.  Since the input gene trees are likely to be noisy
from the error introduced by inferring them from gene sequences, it is
reasonable to set $\alpha$ quite small, which imposes a high standard
for evidence of hybridization.  However, practitioners must decide
(and report) what standards they impose by their choices of $\alpha$
and $\beta$.

We note also that there is no reason that $\alpha$ and $\beta$ should
be chosen to have equal values, and we believe appropriate choices of
both will depend upon the level of noise in the data.  In particular,
\emph{a priori} choices of conventional values such as $0.05$ are
likely poor choices.  Investigating the impact of a range of choices
for $\alpha$ and $\beta$ on the final splits graph is a necessary part
of the analysis.  This issue is addressed briefly below through
several examples of simulated and empirical data sets, but we defer
more complete comments to a future paper directed at empiricists.
 
\medskip

The testing framework described here treats any \qc judged
non-tree-like as supporting a 4-cycle and not a $3_2$-cycle.  Using
Proposition \ref{prop:nobad32}, by an assumption of sufficiently long
edges descended from all hybrid nodes, one can rule out the
possibility of non-tree-like $3_2$-cycles, although an empiricist may
prefer not to make such an assumption.  In a future version of NANUQ
we intend to offer a choice of using an additional statistical test
for $3_2$-cycle networks, but this a test will also be nonstandard,
due to the model having a singularity at the crossing of three line
segments (see Figure \ref{fig:simplex} (C)), and thus requires
additional theoretical development.

\subsection*{Quartet distance with unresolved quartets}\label{subsec:unres}

The quartet distance defined for a binary network earlier in this work
required that all quartet networks, after contraction of 2- and
3-cycles, be binary, with positive lengths for all tree
edges. However, in Step 2 of the NANUQ algorithm we include a
hypothesis test for a star tree, to reduce the possibility of
supporting a particular resolved tree or 4-cycle when the \qc is
nearly uniform and gives at best weak evidence as to what the resolved
topology should be.  We thus must explain how we modify the quartet
distance computed in Step 3 to handle unresolved quartets.

To this end, we make a simple extension of Definition \ref{def:rhonet}
for $\rho_{xy}(Q_{xyzw})$. Guided by the results in \cite{Rhodes2019}
on quartet distances for non-binary trees, we set
  $$\rho_{xy}(Q_{xyzw})=1 \text{ if  $\widetilde Q_{xyzw}$ is a star tree.}$$
In particular, this means a star tree is viewed as separating any two distinct taxa on it.

Under the assumption of a binary network, this modification has no
impact on the asymptotic behavior of the algorithm under the NMSC
model, since by Theorem \ref{thm:statcons} below the probability of
rejecting all quartet star trees approaches 1 as the size of the data
set grows.
  
\subsection*{Statistical consistency}

An estimator of a model parameter is said to be statistically
consistent if the probability of inferring the parameter to
arbitrarily small precision from a data set of size $m$ produced in
accord with the model approaches 1 as $m$ approaches infinity.  Since
the NANUQ algorithm depends upon choices of two significance levels,
$\alpha$ and $\beta$, these choices must be taken into account in
formulating an appropriate notion of consistency for it.  As we will
show, because of the assumption that the unknown network is binary,
the value of $0<\beta<1$ will be inconsequential for this notion,
since as $m$ grows the probability of rejecting a quartet star tree
approaches 1 for every choice of four taxa.

In contrast, when a true quartet network is tree-like, then no matter
how large the data set, we expect to reject the null hypothesis that
the corresponding \qc is tree-like approximately 100$\alpha$\% of the
time.  That is, with probability about $\alpha$, the hypothesis test
will incorrectly support a 4-cycle network when the true quartet
network is tree-like.  This behavior is fundamental to the hypothesis
testing framework, and cannot be avoided.

As a consequence, any notion of statistical consistency for NANUQ must
consider sequences of significance levels $\alpha_m \to 0$.  %
We will show the existence of a sequence of levels $\alpha_m$,
dependent on the sample size $m$, so that as $m$ increases the
probability of correctly failing to reject the null hypothesis
(avoiding type I errors at level $\alpha_m$) approaches 1 while at the
same time the probability of correctly rejecting the null hypothesis
(avoiding type II errors) also goes to 1.  The following theorem then
captures the sense in which NANUQ is statistical consistent.

\begin{theorem}\label{thm:statcons} Under the NMSC model on a binary level-1 metric phylogenetic network $N^+$, for 
numerical parameters in which all induced quartet networks with $3_2$-cycles
are tree-like, there exists a sequence $\alpha_1,\alpha_2,\dots$, with $0<\alpha_m<1$  and $\alpha_m\to 0$ 
such that for any $0<\beta<1$ the  NANUQ algorithm with significance levels $\alpha_m$ and $\beta$ on a data set of 
$m$ gene trees will, with probability approaching 1 as $m\to \infty$, infer the binary unrooted phylogenetic network 
associated to $N^+$ by Theorem \ref{thm:splitsgraphinterp}.
\end{theorem}
\begin{proof} 
  It is enough to show that the $\alpha_m$ can be chosen so that with
  probability approaching 1 the quartet distance computed in the NANUQ
  algorithm exactly agrees with the theoretical quartet distance for
  the true network $N^+$.  As suggested above, this will follow from
  showing that as the sample size $m \to \infty$ with probability
  approaching 1, the hypothesis tests performed will 1) reject a star
  tree at level $\beta$, and 2) fail to reject a tree-like quartet
  network when the true one is tree-like, and reject a tree-like
  quartet network when the true one is non-tree-like at level
  $\alpha_m$.

\smallskip

Consider first the hypothesis test for a star tree for a particular
choice of 4 taxa.  The result we need is essentially a standard one,
but we give a full argument as an orientation for the argument for the
second test.  Since the network is binary, the true multinomial
parameter values are $CF = (p_1, \, p_2, \, p_3)$ with $p_i \ne 1/3$
or $0$, and the null hypothesis is $H_0\tc\ CF = (1/3,\, 1/3,\,1/3)$.
The test statistic is $\lambda = -2 (\ell_0 - \ell)$ where $\ell_0$ is
the supremum of the log-likelihood over parameter values in the null
space (here only (1/3,\, 1/3,\, 1/3)), and $\ell$ is the supremum of
the log-likelihood over the full simplex. The statistic $\lambda$ is
asymptotically $\chi^2$-distributed with $2$ degrees of freedom.

A  \qc $(m_1,m_2,m_3)$  for a sample of size $m$ is a multinomial sample from
a distribution with parameters $(p_1, p_2, p_3)$. Then 
\begin{align*}\lambda &=2 \Big(m_1\log \left(\frac{m_1}m\right) +m_2\log \left(\frac{m_2}m\right)\\
&\qquad\qquad\qquad +m_3\log \left(\frac{m_3}m\right)-m\log\left(\frac1{3}\right) \Big)\\
&=m\cdot2 \Big(\frac{m_1}m\log \left(\frac{m_1}m\right) +\frac{m_2}m\log \left(\frac{m_2}m\right)\\ 
&\qquad\qquad\qquad+\frac{m_3}m\log \left(\frac{m_3}m\right)
 -\log\left(\frac1{3}\right) \Big)\\
&=m X_m,
\end{align*}
where $X_m$ is a random variable. By the law of large numbers and the
continuous mapping theorem $X_m$ converges in probability to
$$c=2\,  \big(p_1\log p_1+ p_2\log p_2+ p_3\log p_3 -\log(1/3) \big)>0.$$
Thus for any $\epsilon>0$ there exits an $M$ such that $m>M$ implies $\PP(X_m>c/2)>1-\epsilon$, and consequently,
that $\PP(\lambda>mc/2)>1-\epsilon.$  This means that for any significance level $0<\beta<1$,
the null hypothesis will be rejected for $m$ sufficiently large with probability at least $1 - \epsilon$.
Since $\epsilon$ was arbitrary, as $m\to \infty$ the probability of rejecting the null hypothesis goes to 1. 
Since there are only finitely many 4-taxon subsets, the probability of rejecting that 
any of these are star-like also goes to 1.

\smallskip

Turning now to the hypothesis test for a tree-like quartet network on 4 specific taxa, 
suppose first the true CF is tree-like.  The likelihood ratio statistic is judged using the approximating
distribution (dependent on the sample size $m$) of the random variable $W_m=W$ described in Theorem 
3.1 of \cite{Mitchell2018}. 
Since the true network is binary, from results in that paper $W_m$ has a limiting distribution as 
$m\to\infty$, which is $\chi^2_1$.   To ensure that the probability of failing to reject the null hypothesis 
approaches 1 as $m\to \infty$, it is enough to choose any sequence of significance levels  with $\alpha_m\to 0$.

In contrast, if the true CF is non-tree-like, we must pick significance levels more carefully.
Without loss of generality, suppose the true CF is $(p_1,p_2,p_3)$ with $p_1\ge p_2>p_3$. 
A \qc $(m_1,m_2,m_3)$  for a sample of size $m$, with $m_{\max}=\max(m_i)$, yields a likelihood ratio statistic 
\begin{align*}
\lambda &=2 \, \Big(m_1\log \left(\frac{m_1}m\right) +m_2\log \left(\frac{m_2}m\right)\\
&\qquad\qquad +m_3\log \left(\frac{m_3}m\right) -  m_{\max}\log\left(\frac{m_{\max}}m\right)\\
&\qquad \qquad -(m-m_{\max})\log\left(\frac{m-m_{\max}}{2m}\right)\Big )\\
&=m\cdot2 \, \Big(\frac{m_1}m\log \left(\frac{m_1}m\right) +\frac{m_2}m\log \left(\frac{m_2}m\right)\\
&\qquad\qquad +\frac{m_3}m\log \left(\frac{m_3}m\right) -  \frac{m_{\max}}m\log\left(\frac{m_{\max}}m\right)\\
&\qquad \qquad -\left(\frac{m-m_{\max}}m\right)\log\left(\frac{m-m_{\max}}{2m}\right) \Big)\\
&=m \, Y_m.
\end{align*}
where $Y_m$ is a random variable. But $Y_m$ converges in probability to
$$d=2( p_2\log p_2 +p_3\log p_3 -(p_2+p_3)\log((p_2+p_3)/2) ) >0.$$
Thus for any $\epsilon>0$ there exits an $M$ such that $m>M$ implies
$\PP(Y_m>d/2)>1-\epsilon$, and thus that
$\PP(\lambda>md/2)>1-\epsilon.$ Let $\alpha_m'=\PP(W_m>md/2)$.  Then
we have that for any $\epsilon>0$ there exists an $M$ such that for
$m>M$ the probability of rejecting the null hypothesis at level
$\alpha_m'$ is $>1-\epsilon$. Thus as $m\to \infty$ the probability of
rejecting the null hypothesis goes to 1.  As the $W_m$ converge in
distribution to a $\chi^2_1$, one also sees that $\alpha_m'\to 0$.

Since there are a finite number of non-tree-like subsets of 4 taxa, we
choose $\alpha_m$ to be the minimum of the $\alpha_m'$ for these
subsets, to ensure the probability of rejecting the null hypothesis
for all of them goes to 1 as $m \to \infty$. As $\alpha_m\to 0$, this
sequence has all the desired properties.
\end{proof}

Note that the assumption in the theorem that all  $3_2$-cyles are tree-like can be ensured through, 
for example, Proposition \ref{prop:nobad32}, by requiring
that no edges descending from hybrid nodes have length less than $\log(5/4)$. 

\smallskip

Although we do not give a formal proof here, NANUQ remains
statistically consistent even in the absence of incomplete lineage
sorting. Informally, one can ``turn off'' ILS in the multispecies
coalescent model by shrinking all population sizes on the species
network. Equivalently, if the species network's branch lengths,
measured in coalescent units, go to $\infty$, then the distribution of
rooted topological gene trees approaches that of a hybridization model
with no ILS. One can thus establish consistency either by taking
appropriate limits in the argument above, or by analyzing quartet
concordance factors for the pure hybridization model directly.

\subsection*{Variants of NANUQ}
The NANUQ algorithm can be adapted to use any means of determining
from data what 4-taxon species network is supported.  Thus future
developments might allow for the replacement of Steps 1 and 2 by
alternative approaches. For instance, one might adopt for analyses of
4-taxon networks an invariants-based approach such as in \cite{Chifman2019},
so that the data becomes aligned genomic sequences. Alternatively,
generalizations of ideas from \cite{Allman2018} which are now being
investigated may allow for determination of rooted triple networks
from genomic sequences, and a rooted triple distance can replace the
quartet distance used here. The essence of the NANUQ approach is to
use a quartet (or rooted triple) distance appropriate to networks
along with the NeighborNet and \CNA, though how one obtains the
information necessary to compute the distance may vary.

\subsection*{Sources of Error}

While NANUQ is a statistically consistent (in the precise sense of
Theorem \ref{thm:statcons}) method of inferring certain network
features from a collection of gene trees produced by the NMSC model,
in practice it must be applied to a finite set of inferred gene trees.
Possible sources of errors in conclusions drawn from NANUQ include:

\begin{enumerate}
\item[1.] Error in gene trees, due to their inference from sequence data,
\item[2.] Small sample size (\emph{e.g.}, few gene trees, many missing taxa on gene trees),
\item[3.] Miscalls of evidence for/against hybridization in individual quartets, in Step 2,
\item[4.] The NeighborNet Algorithm's projection of the split system onto a circular one, in Step 4,
\item[5.] The presence of non-tree-like $3_2$-cycles on some induced quartet networks,
\item[6.] NMSC model misspecification due to any of:
\begin{enumerate}
\item a non-level-1 network,
\item  structure within populations,
\item continuous gene flow between populations.
\end{enumerate}
\end{enumerate}

Thus one should not expect empirical data to necessarily lead to a
splits graph exactly conforming to form described by Theorem
\ref{thm:splitsgraphinterp}.

\smallskip

Note that the algorithm of \cite{Roch2018} offers an alternative to
NeighborNet that might reduce the error arising in passing to a
circular split system from the quartet distance. However, this has not
been implemented in general purpose software yet, so we were unable to
test its performance.

We have chosen not to suggest any automatic interpretation of the
output of NANUQ, such as a mechanism for producing the closest splits
graph (by some measure) that conforms exactly to the form described by
Theorem \ref{thm:splitsgraphinterp}. Thus the user must visually
consider the output, which will reflect some of the error.  In
particular, SplitsTree offers a capability of removing splits with
small weight from a splits graph, and this can be useful for removing
some of the noise remaining after projecting onto a circular split
system.


\section*{Examples}\label{sec:examples}

In this section we present three examples of data analysis with
NANUQ. The first uses a simulated data set of gene trees (without any
gene tree inference error), the second the well-known and well-studied
yeast data set of \cite{Rokas2003}, and the third the butterfly data
set of \cite{Martin2013}.  For the empirical data sets, we use gene
trees previously inferred from genetic sequences.  Reported running
times are from a Macbook Pro computer with a 3.1 GHz processor.

\begin{example}\label{ex:sim}

\begin{table*}
  \caption{The metric species network $N^+$, in extended Newick format, used for simulating gene trees under the NMSC model. The topology of $N^+$ is shown in Figure \ref{fig:nanuq}.}\label{tab:newick}
%
%
%
%
%
%
(((((a\tc1.5,(((b\tc.8,h1\#.5\tc.1)x1\tc.2,(c\tc.7)h1\#.5\tc.3)x2\tc.3)h2\#.5\tc.2)x3\tc1.5,(h2\#.5\tc.2,d\tc1.5)x4\tc1.5)x5\tc2,h3\#.5\tc1.5)x6\tc0.5,

(((e\tc2,(f\tc1,((g\tc.25,h\tc.25)x7\tc.25)h4\#.5\tc.5)x8\tc1)x9\tc1,(h4\#.5\tc.5,i\tc1)x10\tc2)x11\tc0.5)h3\#.5\tc2)x12\tc1,((((j\tc4.5,(k\tc3.5,((l\tc2.75,m\tc2.75)x13\tc.25)

h5\#.3\tc.5)x14\tc1)x15\tc1,((((n\tc1)h6\#.5\tc2,h6\#.5\tc2)x16\tc.5,h5\#.3\tc.5)x17\tc1,(o\tc3.5,p\tc3.5)x18\tc1)x19\tc1)x20\tc.25)h7\#.5\tc.5,h7\#.5\tc.5)x21\tc.25)r;
\end{table*}

We generated a data set of 1000 gene trees using Hybrid-Lambda \cite{Hybrid-lambda} on the 
species network $N^+$ shown in Figure \ref{fig:nanuq}, with branch lengths in coalescent units and 
hybridization parameters as shown in Table \ref{tab:newick}.

In running NANUQ on this data set, our implementation of Steps 1-3 in R
required about 63\emph{s} of computation time. We considered a range
of values of $\alpha$ and $\beta$ for the hypothesis tests. To
visualize outcomes of the hypothesis tests, we produced simplex plots
such as those shown in Figure \ref{fig:simsimplex}, which plot
empirical CFs (i.e., \qcs normalized to sum to 1) for each set of 4
taxa, color coded to indicate test outcomes.  The results of the
hypothesis tests gave a rather clean separation of empirical CFs into
those close to the 3 line segments which were classified as tree-like,
and those farther away which were viewed as supporting a 4-cycle.  We found
that for any level $\alpha$ in the range $10^{-17}\leq\alpha\leq .01$,
our hypothesis tests drew the same conclusions as to which \qcs
supported a 4-cycle (red triangles). The close clustering of the \qcs
not rejected as tree-like (blue circles) around the tree model also
suggests little error in them, so that a rather large value of $\beta$
might be sufficient to test for lack of resolution. When $\beta$ is
set to .05, all \qcs result in rejection of the star tree hypothesis.
As shown in the figure, when $\beta$ is reduced to $10^{-19}$ a single
failure to reject the star tree hypothesis occurs (tan square).  Using $\alpha=.01$
and $ \beta=.05$ to compute the quartet distance, from SplitsTree4 we
obtain the splits graph on the right of Figure \ref{fig:nanuq}.  Under
the rules of Theorem \ref{thm:splitsgraphinterp}, this correctly gives
all features of $N^-$ inferable by NANUQ.

\begin{figure*}\begin{center}
\includegraphics[scale=.7]{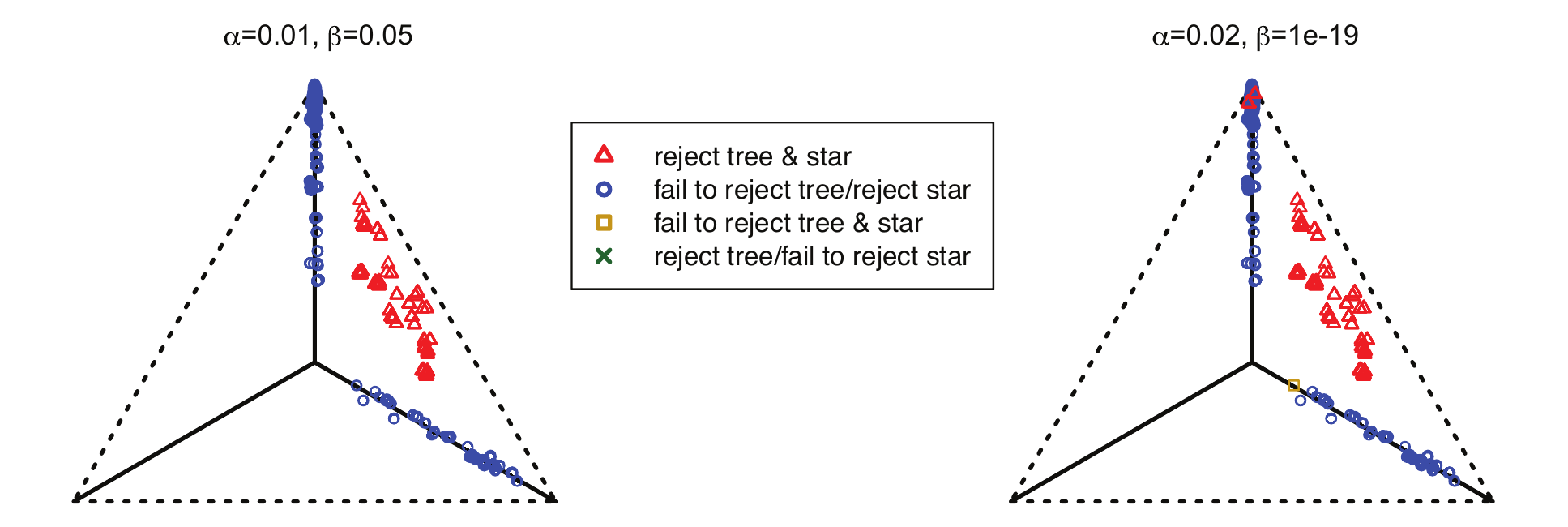}	
\caption{Representative simplex plots for empirical CFs, with hypothesis testing results, computed
from a simulated data set of 1000 gene trees from the species network given in Table \ref{tab:newick}. }
 \label{fig:simsimplex} 
\end{center}
\end{figure*}

Reducing the sample size to 300 gene trees, while using the same values of $\alpha$ and $\beta$, we 
obtained the same correct inference result.
\end{example}

\begin{example}\label{ex:yeast}
For the second example we use a subset of the yeast data set of \cite{Rokas2003}, which has been analyzed by multiple investigators \cite{Bloomquist2010,Holland2004,Wen2018,Wu2008,Yu2011,Zhang2018}.
It consists of 106 gene trees, each with a single allele sampled from seven 
\emph{Saccharomyces} species: \emph{S. cerevisiae} (S\,cer), \emph{S. paradoxus} (S\,par), 
\emph{S. mikatae} (S\,mik), \emph{S. kudriavzevii} (S\,kud), 
\emph{S. bayanus} (S\,bay), \emph{S. castellii} (S\,cas), \emph{S. kluyveri} (S\,klu), and the 
outgroup fungus \emph{Candida albicans} (C\,alb). Running time for NANUQ's Steps 1-3 was 
under 0.5\emph{s}.

\begin{figure*}
\begin{center}
	\includegraphics[scale=.7]{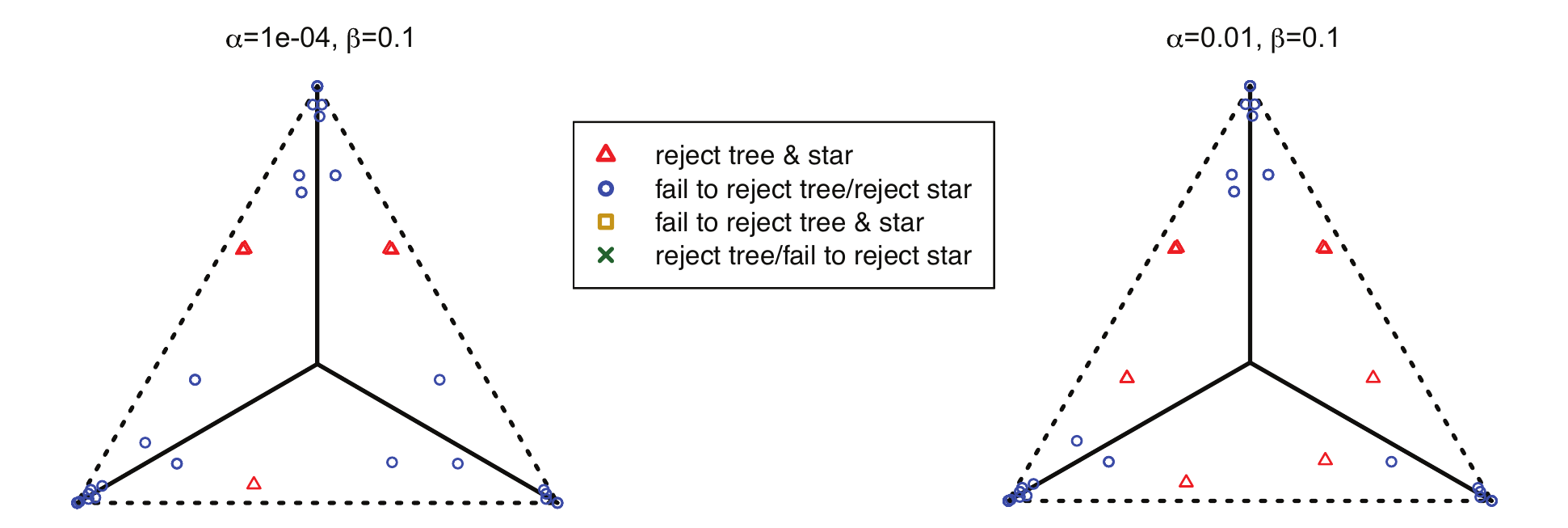} 
	\caption{Simplex plots for hypothesis test results on the
          yeast data set, with two choices of significance levels
          $\alpha=10^{-4}$ and $10^{-2}$ with $\beta=0.1$. The choice
          of $\beta$ here is largely irrelevant, as no plotted
          empirical CFs are near the center. Larger $\alpha$ results
          in more empirical CFs being determined as supporting
          4-cycles, as several blue circles on the left
          change to red triangles on the right.}
	\label{fig:exYeastsimplex}
\end{center}
\end{figure*}

Displayed in Figure \ref{fig:exYeastsimplex} are some sample results
from hypothesis tests for several choices of $\alpha$ and $\beta$.  As
all of the empirical CFs are far from $(1/3, \, 1/3, \, 1/3)$, the CF
for the star tree, only a quite small $\beta$ would lead to failing to
reject the star tree for any set of 4 taxa. Thus, for this data set,
we set $\beta=0.1$ and classify all quartet networks as resolved,
either as trees or 4-cycle networks. (We also see that no empirical
CFs are plotted near the locations of non-tree-like $3_2$-cycle CFs,
giving us some confidence in NANUQ's assumption that there are none in
the data.)  We chose values of $\alpha=10^{-4}$ and $10^{-2}$ as the
first of these results in only the most extreme empirical CFs (far
from the tree-like CF line segments) being interpreted as supporting
4-cycle networks, while the larger value, in imposing a less strict
standard for evidence of hybridization, classifies more of those
empirical CFs distant from the null model as 4-cycles.  Further
increasing $\alpha$ to values $>.08$ would result in additional
classification of 4-cycle networks, but we chose to interpret those
deviations from tree-like-ness as due to stochastic (or other) noise.

\begin{figure*}\begin{center}
		\includegraphics[width=8.1cm]{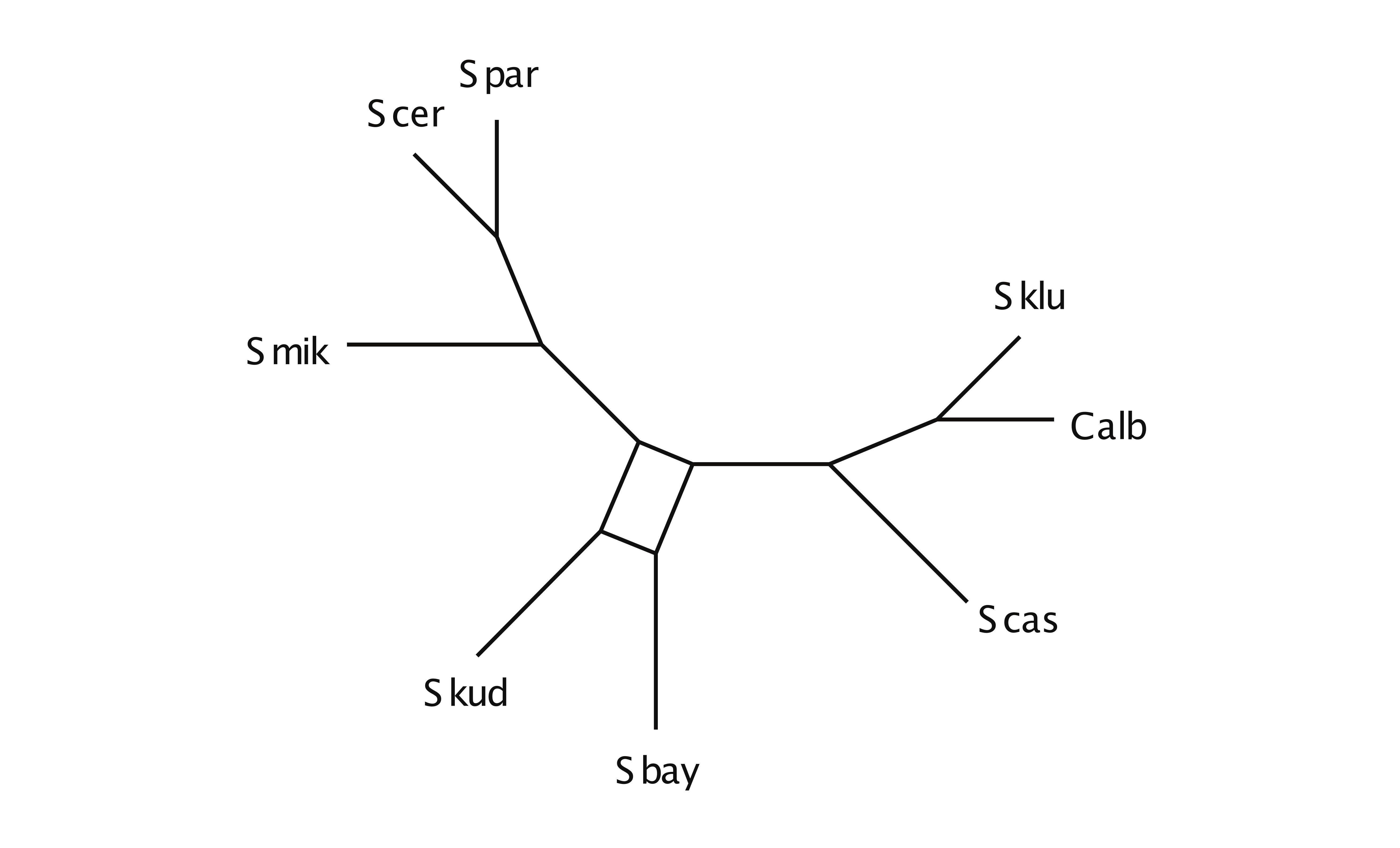} 
		\includegraphics[width=8.1cm]{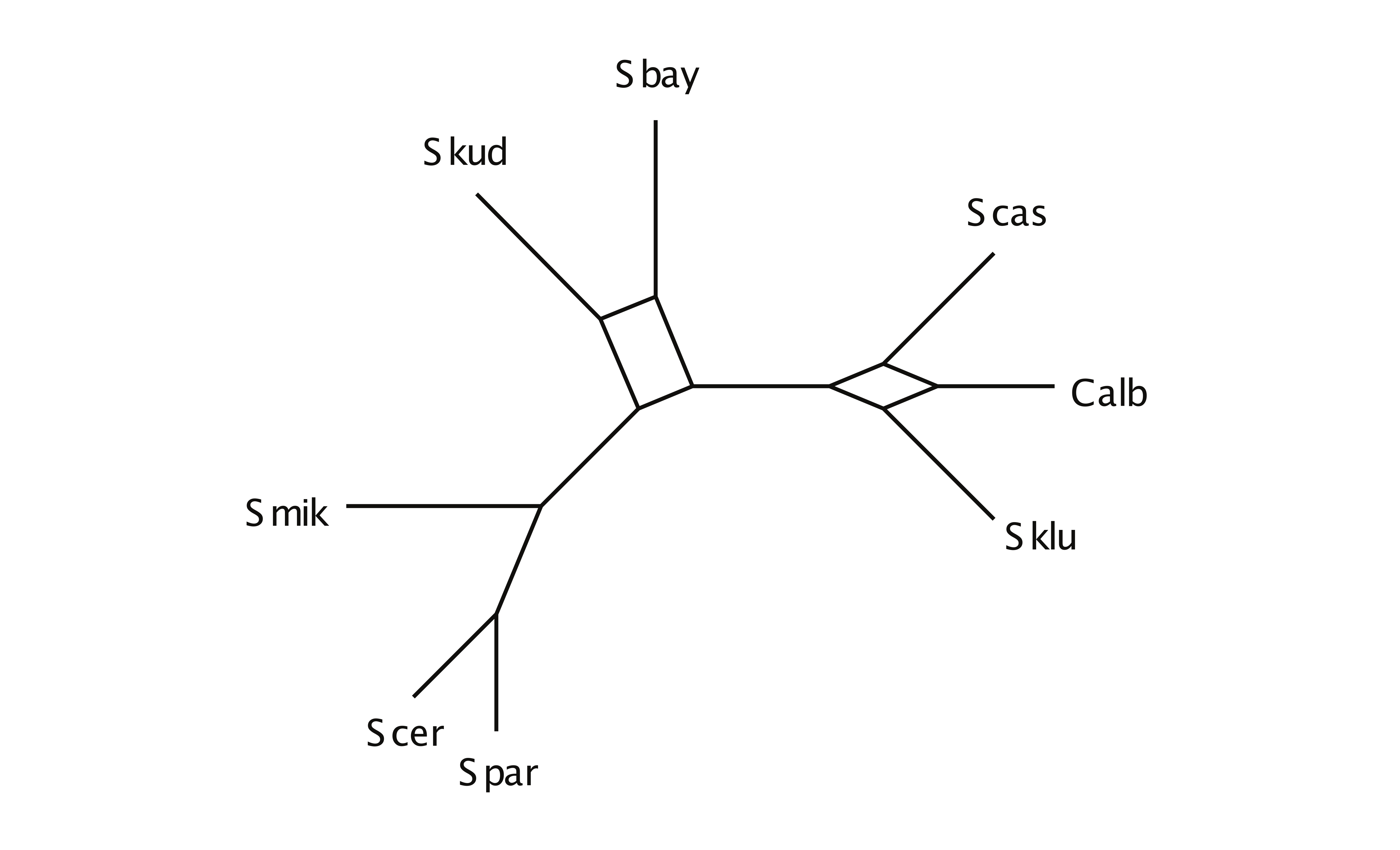}
		\caption{Networks inferred by NANUQ for yeast data  of Example \ref{ex:yeast} with $\beta=0.1$ and
		$\alpha=10^{-4}$ (L) or $10^{-2}$  (R) . }\label{fig:exYeastnanuq}
\end{center}\end{figure*}

For each of the choices of $\alpha$, $\beta$, the splits graphs produced in NANUQ's use of SplitsTree are shown in 
Figure \ref{fig:exYeastnanuq}. Since these show only 4-cycles, they can be directly interpreted as indicating the 
undirected version of the true level-1 network topology relating the taxa, with all 2- and 3-cycles contracted.   
We obtain no information on root location from NANUQ since no cycles have size larger than 4.

\end{example}

\begin{figure*}\begin{center}
\includegraphics[scale=.7]{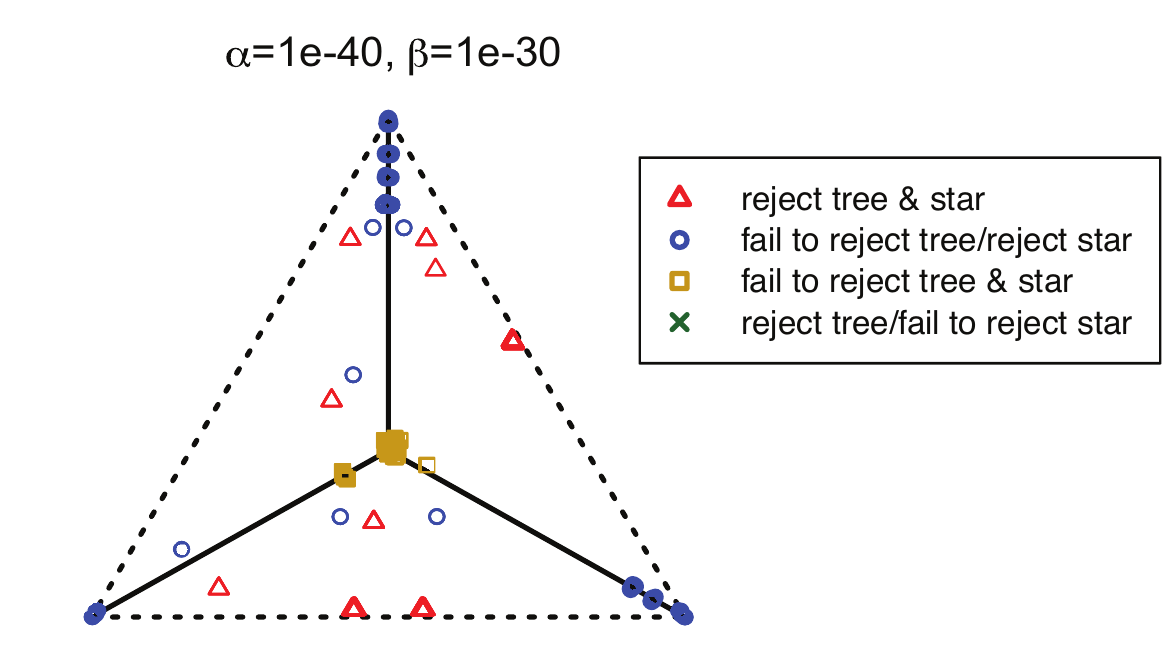} \caption{ Simplex plot showing hypothesis test results 
for the \emph{Heliconius} data set of Example \ref{ex:butter}.  } 
\label{fig:exBsimplex} 
\end{center}\end{figure*}

\begin{example}\label{ex:butter}
  For the third example we use a \emph{Heliconius} butterfly data
  set \cite{Martin2013}, also analyzed in \cite{Chifman2019}, which
  has been presented as evidence of gene flow between sympatric species.  This data set
  consists of 2909 loci, derived from non-overlapping 100-kb windows
  in the full genome of individuals.  Four individuals were sampled
  from three ingroup \emph{Heliconius} species: \emph{H.~rosina},
  \emph{H.~melpomene}, \emph{H.~cydno} (labelled chioneus), and one
  individual from four outgroup species \emph{H.~ethilia},
  \emph{H.~hecale}, \emph{H.~p.~sergestus}, and \emph{H.~pardalinus}.

Running time for Steps 1-3 of the algorithm was about 174\emph{s}.
Figure \ref{fig:exBsimplex} shows results of hypothesis tests for one choice of $\alpha$ and $\beta$, 
with Figure \ref{fig:exBnanuq} the resulting splits graph and inferred network structure. 
Note that a number of the empirical CFs (tan squares in Figure \ref{fig:exBsimplex}) are close to the star-tree CF, 
and the choice of test level $\beta = 10^{-30}$ results in these being treated as unresolved quartets, 
giving the multifurcations in the splits graph for the three multi-sampled taxa. If $\beta$ is made larger so that
star trees are rejected more often, then blobs can appear within the single taxon groups. 
For a broad range of choices for $\alpha$ and $\beta$ (not shown), the three ingroups
\emph{H.~rosina, H.~melpomene, H.~c.~chioneus} and the outgroup are related by
a 4-cycle by NANUQ.

Though not the taxa of focus in the study \cite{Martin2013}, the splits graph of Figure \ref{fig:exBnanuq} depicts 
interesting relationships between the outgroup taxa and illustrates the flexibility of our
analysis.  While difficult to see in the SplitsTree4 output, there is a split with very small weight 
separating  \emph{H.~ethilia}, \emph{H.~sergestus}, and \emph{H.~pardalinus} 
from the rest of the taxa.   SplitsTree4 allows such small weight splits to be filtered out, and doing so 
leaves a 5-dart pointed at \emph{H.~sergestus}.  
However, for different values of $\alpha$ the $5$-dart can change: for example,
for $\alpha=10^{-17}$ the $5$-dart points to \emph{H.~ethilla} instead.
Thus while the central 4-cycle is very well supported, across many values of 
$\alpha$ and $\beta$, one might not want to draw firm conclusions on other hybridizations
in this data set. The analysis does, however, suggest that the relationships between these taxa might warrant further investigation.

\begin{figure*}
 \includegraphics[scale=.58]{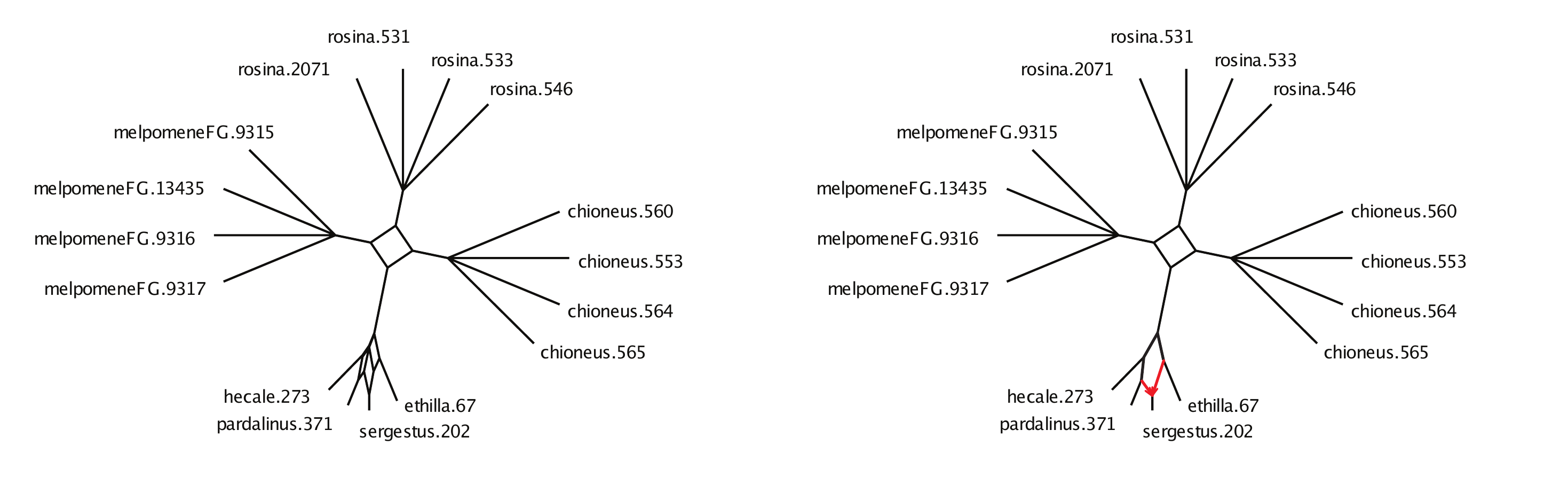}  \hfill \ 
 \caption{(L) Splits graph for \emph{Heliconius} data set of Example \ref{ex:butter}, 
for $\alpha=10^{-40}$, $\beta=10^{-{30}}$, and (R) NANUQ inferred network structure.}
\label{fig:exBnanuq}
\end{figure*}

\end{example}
 

\begin{backmatter}

\section*{Competing interests}
  The authors declare that they have no competing interests.

\section*{Author's contributions}
    All authors contributed equally in the development of the theory, implementation of the algorithm, and
    manuscript writing. They all approved the final manuscript. 

\section*{Author details}
  Department of of Mathematics and Statistics, University of Alaska Fairbanks, 
  1792 Ambler Lane, Fairbanks, AK, 99775  USA
      
\section*{Acknowledgements}
 This research was supported, in part, by the National Institutes of Health Grant R01 GM117590, awarded
 under the Joint DMS/NIGMS Initiative to Support Research at the Interface of the Biological and
 Mathematical Sciences.


\bibliographystyle{bmc-mathphys} 
\bibliography{Hybridization}      

\end{backmatter}
\end{document}